\RequirePackage{luatex85}

\documentclass[11pt,letterpaper]{article}
\usepackage[letterpaper, margin=1in]{geometry}

\usepackage{iftex}

\ifluatex
\else
\ifxetex
\else
\usepackage[utf8]{inputenc}
\usepackage[T1]{fontenc}
\fi
\fi

\usepackage{amsmath}
\usepackage{amssymb}
\usepackage{amsthm}
\usepackage{mathtools}
\usepackage{thmtools}

\ifluatex
\usepackage{fontspec}
\usepackage{unicode-math}
\else
\ifxetex
\usepackage{fontspec}
\usepackage{unicode-math}
\else
\usepackage{isomath}
\usepackage{libertine}
\usepackage[libertine]{newtxmath} 
\usepackage[scaled=0.96]{zi4} 
\fi
\fi

\usepackage{microtype}

\usepackage[ruled,vlined,linesnumbered]{algorithm2e}

\usepackage{xcolor}

\usepackage[textwidth=2cm,textsize=footnotesize]{todonotes}

\newcommand{\remove}[1]{}

\usepackage{lineno}

\usepackage[
	style=alphabetic,
    maxalphanames=4,
    minalphanames=4,
    maxcitenames=4,
    minnames=1,
    maxbibnames=20,
	backref=true,
	doi=true,
	url=true,
	backend=bibtex8,
]{biblatex}

\usepackage[ocgcolorlinks]{hyperref} 

\newcommand*\samethanks[1][\value{footnote}]{\footnotemark[#1]}

\allowdisplaybreaks[1]

\makeatletter
\g@addto@macro\bfseries{\boldmath}
\makeatother

\makeatletter
\g@addto@macro\mdseries{\unboldmath}
\g@addto@macro\normalfont{\unboldmath}
\g@addto@macro\rmfamily{\unboldmath}
\g@addto@macro\upshape{\unboldmath}
\makeatother

\DeclareCiteCommand{\citem}
    {}
    {\mkbibbrackets{\bibhyperref{\usebibmacro{postnote}}}}
    {\multicitedelim}
    {}

\renewcommand*{\multicitedelim}{\addcomma\space}

\AtEveryCitekey{\clearfield{note}}

\makeatletter
\AtBeginDocument{%
    \newlength{\temp@x}%
    \newlength{\temp@y}%
    \newlength{\temp@w}%
    \newlength{\temp@h}%
    \def\my@coords#1#2#3#4{%
      \setlength{\temp@x}{#1}%
      \setlength{\temp@y}{#2}%
      \setlength{\temp@w}{#3}%
      \setlength{\temp@h}{#4}%
      \adjustlengths{}%
      \my@pdfliteral{\strip@pt\temp@x\space\strip@pt\temp@y\space\strip@pt\temp@w\space\strip@pt\temp@h\space re}}%
    \ifpdf
      \typeout{In PDF mode}%
      \def\my@pdfliteral#1{\pdfliteral page{#1}}
      \def\adjustlengths{}%
    \fi
    \ifxetex
      \def\my@pdfliteral #1{}
      \def\adjustlengths{\setlength{\temp@h}{-\temp@h}\addtolength{\temp@y}{1in}\addtolength{\temp@x}{-1in}}%
    \fi%
    \def\Hy@colorlink#1{%
      \begingroup
        \ifHy@ocgcolorlinks
          \def\Hy@ocgcolor{#1}%
          \my@pdfliteral{q}%
          \my@pdfliteral{7 Tr}
        \else
          \HyColor@UseColor#1%
        \fi
    }%
    \def\Hy@endcolorlink{%
      \ifHy@ocgcolorlinks%
        \my@pdfliteral{/OC/OCPrint BDC}%
        \my@coords{0pt}{0pt}{\pdfpagewidth}{\pdfpageheight}%
        \my@pdfliteral{F}
        \my@pdfliteral{EMC/OC/OCView BDC}%
        \begingroup%
          \expandafter\HyColor@UseColor\Hy@ocgcolor%
          \my@coords{0pt}{0pt}{\pdfpagewidth}{\pdfpageheight}%
          \my@pdfliteral{F}
        \endgroup%
        \my@pdfliteral{EMC}%
        \my@pdfliteral{0 Tr}
        \my@pdfliteral{Q}%
      \fi
      \endgroup
    }%
}
\makeatother

\usepackage{enumerate}

\colorlet{DarkRed}{red!50!black}
\colorlet{DarkGreen}{green!50!black}
\colorlet{DarkBlue}{blue!50!black}

\hypersetup{
	linkcolor = DarkRed,
	citecolor = DarkGreen,
	urlcolor = DarkBlue,
	bookmarksnumbered = true,
	linktocpage = true
}

\DontPrintSemicolon
\SetKwProg{Procedure}{Procedure}{}{}
\SetProcNameSty{textsc}
\SetFuncSty{textsc}
\SetKw{KwAnd}{and}
\SetKw{KwDo}{do}

\SetCommentSty{mycommfont}

\declaretheorem[numberwithin=section]{theorem}
\declaretheorem[numberlike=theorem]{lemma}

\declaretheorem[numberlike=theorem]{question}

\DeclareMathOperator{\poly}{poly}

\SetKwFunction{Hit}{HittingSet}
\SetKwFunction{PreProcessing}{PreProcessing}
\SetKwFunction{Query}{Query}

\makeatletter
\renewcommand{\paragraph}{%
  \@startsection{paragraph}{4}%
  {\z@}{1.4ex \@plus 1ex \@minus .2ex}{-1em}%
  {\normalfont\normalsize\bfseries}%
}
\makeatother

\title{Fast $2$-Approximate All-Pairs Shortest Paths}

\author{Michal Dory\thanks{University of Haifa} \and Sebastian Forster\thanks{Department of Computer Science, University of Salzburg. This work is supported by the Austrian Science Fund (FWF): P 32863-N. This project has received funding from the European Research Council (ERC) under the European Union's Horizon 2020 research and innovation programme (grant agreement No~947702).} \and Yael Kirkpatrick\thanks{MIT, this material is based upon work supported by the National Science Foundation Graduate Research Fellowship under Grant No 2141064} \and Yasamin Nazari\thanks{VU Amsterdam. This work was partially conducted while the author was a postdoc at University of Salzburg.} \and Virginia Vassilevska Williams \thanks{MIT, Supported by NSF Grants CCF-2129139 and CCF-2330048 and BSF Grant 2020356.} \and Tijn de Vos\samethanks[2]}

\date{}

\hypersetup{
	pdftitle = {Fast 2-Approximate All-Pairs Shortest Paths},
	pdfauthor = {Michal Dory, Sebastian Forster, Yael Kirkpatrick, Yasamin Nazari, Virginia Vassilevska Williams, Tijn de Vos}
}

\begin{document}
\begin{titlepage}
    \maketitle
    \begin{abstract}
    In this paper, we revisit the classic approximate All-Pairs Shortest Paths (APSP) problem in undirected graphs. 
For unweighted graphs, we provide an algorithm for $2$-approximate APSP in $\tilde O(n^{2.5-r}+n^{\omega(r)})$ time, for any $r\in[0,1]$. This is $O(n^{2.032})$ time, using known bounds for rectangular matrix multiplication~$n^{\omega(r)}$~[Le Gall, Urrutia, SODA 2018]. Our result improves on the $\tilde{O}(n^{2.25})$ bound of [Roditty, STOC 2023], and on the $\tilde{O}(m\sqrt n+n^2)$ bound of [Baswana, Kavitha, SICOMP 2010] for graphs with $m\geq n^{1.532}$ edges. 

For weighted graphs, we obtain $(2+\epsilon)$-approximate APSP in $\tilde O(n^{3-r}+n^{\omega(r)})$ time, for any $r\in [0,1]$. This is $O(n^{2.214})$ time using known bounds for $\omega(r)$. It improves on the state of the art bound of $O(n^{2.25})$ by [Kavitha, Algorithmica 2012]. Our techniques further lead to improved bounds in a wide range of density for weighted graphs. 
In particular, for the sparse regime we construct a distance oracle in $\tilde O(mn^{2/3})$ time that supports $2$-approximate queries in constant time. For sparse graphs, the preprocessing time of the algorithm matches conditional lower bounds [Patrascu, Roditty, Thorup, FOCS 2012; Abboud, Bringmann, Fischer, STOC 2023]. To the best of our knowledge, this is the first 2-approximate distance oracle that has subquadratic preprocessing time in sparse graphs.

We also obtain new bounds in the near additive regime for unweighted graphs. We give faster algorithms for $(1+\epsilon,k)$-approximate APSP, for $k=2,4,6,8$.

We obtain these results by incorporating fast rectangular matrix multiplications into various combinatorial algorithms that carefully balance out distance computation on layers of sparse graphs preserving certain distance information.

    \end{abstract}
    \thispagestyle{empty}

\newpage
\tableofcontents
\thispagestyle{empty}
\end{titlepage}

\newpage
\newpage
\section{Introduction}
The All-Pairs Shortest Paths (APSP) problem is one of the most fundamental problems in graph algorithms. In this problem, the goal is to compute the distances between all pairs of vertices in a graph. It is well-known that APSP can be solved in $O(n^3)$ time in directed weighted graphs with $n$ vertices using the Floyd-Warshall algorithm (see \cite{leiserson1994introduction}), or in $\tilde{O}(mn)$ time using Dijkstra's algorithm, where $m$ is the number of edges in the graph (see also \cite{thorup1999undirected,pettie2004new,pettie2005shortest}). A slightly subcubic algorithm for APSP with running time $n^{3}/2^{\Omega(\log{n})^{1/2}}$ was given by Williams \cite{williams2014faster}. 
A natural hypothesis in graph algorithms (see \cite{RodittyZ04,vsurvey}) is that $n^{3-o(1)}$ time is required to solve APSP in weighted graphs; this is known as the APSP hypothesis.
Subcubic equivalences\footnote{Two problems are subcubically equivalent if one problem can be solved in time $O(n^{3-\epsilon})$ for some $\epsilon >0$ if and only if the other can be solved in time $O(n^{3-\delta})$ for some $\delta > 0$.} between the APSP problem and many other problems such as finding a negative weight triangle or finding the radius of the graph \cite{vassilevska2015hardness,williams2018subcubic,DBLP:journals/talg/AbboudGW23} significantly strengthened the belief in the APSP hypothesis.

All the above algorithms solve the problem in weighted, directed graphs. If the graphs are unweighted and undirected, APSP can be solved faster, in $\tilde{O}(n^{\omega})$ time, using fast matrix multiplication \cite{seidel1995all,galil1997all}, where $\omega<2.372$ is the exponent of square matrix multiplication \cite{coppersmith1987matrix,williams2012multiplying,le2014powers,alman2021refined,DBLP:journals/corr/abs-2210-10173}. For weighted, directed graphs with bounded weights, Zwick showed an algorithm that takes $O(n^{2.529})$ time \cite{Zwick02,GallU18}\footnote{The running time of this algorithm depends on the exponent of rectangular matrix multiplication, and becomes $O(n^{2.529})$ with the rectangular matrix multiplication algorithm from \cite{GallU18}.}. 
In $\tilde{O}((n^{\omega}/\epsilon) \cdot \log{W})$ time it is also possible to obtain a $(1+\epsilon)$-approximation for APSP in weighted directed graphs \cite{zwick1998all}, where $W$ is the maximum edge weight and the weights are scaled such that the smallest non-zero weight is 1. 

However, the above complexities can be high for large graphs, and it is desirable to have faster algorithms. While it may be difficult to get faster algorithms for exact APSP, a natural question is whether we can get faster \emph{approximation} algorithms for APSP. We say that an algorithm gives an $(\alpha,\beta)$-approximation for APSP if for any pair of vertices $u,v$ it returns an estimate $\delta(u,v)$ of the distance between $u$ and $v$ such that $d(u,v) \leq \delta(u,v) \leq \alpha \cdot d(u,v) + \beta$, where $d(u,v)$ is the distance between $u$ and $v$. If $\beta=0$, we get a purely multiplicative approximation that we refer to as $\alpha$-approximation, $\alpha$ is sometimes called the \emph{stretch} of the algorithm. If $\alpha=1$, we get a purely additive approximation, and call it a $+\beta$-approximation. 

Dor, Halperin, and Zwick \cite{DHZ00} showed that obtaining a $(2-\epsilon)$-approximation for APSP is at least as hard as Boolean matrix multiplication. This implies that we cannot get algorithms with running time below $O(n^{\omega})$ for approximate APSP with approximation ratios below 2, as it would lead to algorithms for matrix multiplication with the same running time. Their reduction holds even in the case that the graphs are unweighted and undirected. In the directed case the same result holds for \emph{any} approximation. This makes the special case of a 2-approximation in undirected graphs a very interesting special case, as this is the first approximation ratio where we can beat the $O(n^{\omega})$ bound.
Since in directed graphs any approximation for APSP requires $\Omega(n^{\omega})$ time, we will focus from now on \emph{undirected} graphs. 

\paragraph{2-Approximate APSP.}
While 2-approximation algorithms for APSP have been extensively studied \cite{AingworthCIM99,DHZ00,CohenZ01,BaswanaK10,Kavitha12,DengKRWZ22,durr2023improved,Roditty23} (see \autoref{table_apsp} for a summary), it is still unclear what the best running time is that can be obtained for this problem. 
This question is also open even in the simpler case of unweighted and undirected graphs.
The study of 2-approximate APSP was treated by the seminal work of Aingworth, Chekuri, Indyk, and Motwani~\cite{AingworthCIM99} that showed an additive +2-approximation algorithm for APSP that takes $\tilde{O}(n^{2.5})$ time. Their algorithm works in undirected unweighted graphs. 
While the running time of the algorithm is above $O(n^{\omega})$, it is a simple \emph{combinatorial} algorithm, where we use the widespread informal terminology that an algorithm is combinatorial if it does not use fast matrix multiplication techniques. 
The running time was later improved by Dor, Halperin, and Zwick that showed a $+2$-approximation algorithm in $\tilde{O}(\min\{n^{3/2}m^{1/2}, n^{7/3} \})$ time \cite{DHZ00}. This algorithm is still the fastest known combinatorial algorithm for a +2-approximation. Very recent results show that using fast matrix multiplication techniques, one can get faster algorithms for a +2-approximation, leading to an $O(n^{2.260})$ running time \cite{DengKRWZ22,durr2023improved}. 
All the above mentioned algorithms give an additive +2-approximation in unweighted undirected graphs, which also implies a multiplicative 2-approximation. This holds since for pairs $u,v$ where $d(u,v)=1$ we can learn the exact distance in $O(m)$ time, and once $d(u,v) \geq2$, an additive +2-approximation is also a multiplicative 2-approximation. However, if our goal is to obtain a multiplicative approximation, it may be possible to get a faster algorithm.

Algorithms for multiplicative 2-approximate APSP are studied in \cite{CohenZ01,BaswanaK10,Kavitha12,Roditty23}. In particular, Cohen and Zwick \cite{CohenZ01} showed that a 2-approximation for APSP can be computed in $\tilde{O}(n^{3/2}m^{1/2})$ time also in weighted graphs. A faster algorithm with running time of $\tilde{O}(m\sqrt{n}+n^2)$ was shown by Baswana and Kavitha \cite{BaswanaK10}, this algorithm also works in weighted graphs.
While in dense graphs the running time is $\tilde{O}(n^{2.5})$, this algorithm is more efficient in sparse graphs. In particular, for graphs with $O(n^{3/2})$ edges, the running time becomes $\tilde{O}(n^2)$.
A faster algorithm for denser graphs was shown by Kavitha that  showed a $(2+\epsilon)$-approximation for weighted APSP in $\tilde{O}(n^{2.25})$ time~\cite{Kavitha12}.\footnote{In the introduction we assume that $1/\epsilon=n^{o(1)}$ for simplicity of presentation. Similarly, when we discuss weighted graphs, we assume polynomial weights.} This algorithm exploits fast matrix multiplication techniques. Very recently, Roditty showed a combinatorial algorithm for 2-approximate APSP in unweighted undirected graphs in $\tilde{O}(n^{2.25})$ time~\cite{Roditty23}. To conclude, currently the fastest 2-approximation algorithms take $\tilde{O}(\min \{m\sqrt{n}+n^2, n^{2.25} \})$ time, and it is still unclear what is the best running time that can be obtained for the problem. In particular, the following question is still open.  

\begin{question}
Can we get a 2-approximation for APSP in $\tilde{O}(n^2)$ time?
\end{question}

It is worth mentioning that a slightly higher approximation of $(2,1)$ can be obtained in $\tilde{O}(n^2)$ time in unweighted undirected graphs \cite{BaswanaK10,berman2007faster,Sommer16,Knudsen17}. 
In weighted undirected graphs, a multiplicative 3-approximation can be obtained in $\tilde{O}(n^2)$ time \cite{CohenZ01,BaswanaK10}. Furthermore, \cite{BaswanaK10} also gives a $(2,W_{u,v})$-approximation for weighted undirected graphs in $\tilde{O}(n^2)$ time, where $W_{u,v}$ is the maximum weight of an edge in the shortest path between $u$ and $v$. In other words, for any pair of vertices $u,v$ the algorithm returns an estimate $\delta(u,v) \leq 2d(u,v)+W_{u,v}$. 

\paragraph{Distance Oracles.} 
Note that $\Omega(n^2)$ is a natural barrier for APSP, as just writing the output of the problem takes $\Omega(n^2)$ time. However, $\Omega(n^2)$ time or space may be too expensive for large graphs, and if we do not need to write the output explicitly, we may be able to obtain algorithms with \emph{subquadratic} time or space. This has motivated the study of \emph{distance oracles}, which are compact data structures that allow fast query of (possibly approximate) distances between any pair of vertices. The study of near-optimal approximate distance oracles was initiated by the seminal work of Thorup and Zwick \cite{TZ2005} that showed that for any integer $k \geq 2$, a distance oracle of size $O(kn^{1+1/k})$ can be constructed in $O(kmn^{1/k})$ time, such that for any pair of vertices we can obtain a $(2k-1)$-approximation of the distance between them in query time $O(k)$. As an important special case, it gives 3-approximate distance oracles of size $O(n^{3/2})$ in $O( m\sqrt{n}) $ time. The construction is for weighted undirected graphs. Note that the distance oracle uses \emph{subquadratic} space, and the construction time is \emph{subquadratic} for sparse graphs. As shown later these distance oracles can be built also in $\tilde{O} (\min \{kmn^{1/k}, n^2\}$) time \cite{BaswanaK10}. The construction time was further improved in \cite{wulff2012approximate}, where the query time and size were further improved in \cite{wulff2013approximate,chechik2014approximate,chechik2015approximate}.

Distance oracles of size $\tilde{O}(n^{5/3})$ that provide a $(2,1)$-approximation in unweighted undirected graphs are studied in \cite{BaswanaGS09,PatrascuR14,Sommer16,Knudsen17}. This tradeoff between stretch and space is optimal assuming the hardness of set intersection \cite{PatrascuR14,PatrascuRT12}. The fastest of them run in $\min\{\tilde{O}(n^2, mn^{2/3})\}$ time \cite{BaswanaGS09,Sommer16,Knudsen17}.\footnote{The running time of $\tilde{O}(mn^{2/3})$ is implicit in \cite{BaswanaGS09}.}

Recently, slightly subquadratic algorithms with nearly 2-approximations were given for undirected unweighted graphs. 
In particular, a distance oracle with stretch $ (2 (1 + \epsilon), 5)$ and slightly subquadratic running time is given by Akav and Roditty \cite{AkavR20}. In a follow-up work, Chechik and Zhang constructed a $(2,3)$-approximate distance oracle of size $\tilde{O}(n^{5/3})$ in $\tilde{O}(m+n^{1.987})$ time. They also study other trade-offs between the stretch and running time, and in particular show that $(2+\epsilon,c(\epsilon))$-approximate distance oracle of size $\tilde{O}(n^{5/3})$ can be constructed in $\tilde{O}(m+n^{5/3 + \epsilon})$ time, where $c(\epsilon)$ is a constant depending exponentially on $1/\epsilon$. The preprocessing time of this distance oracle nearly matches a recent conditional lower bound by Abboud, Bringmann, and Fischer \cite{aboud2022stronger}, who showed that $m^{5/3-o(1)}$ time is required for a $(2+o(1))$-approximation, conditional on the 3-SUM conjecture.  
While there are subquadratic constructions of distance oracles with nearly 2-approximations, all existing algorithms for 2-approximations take at least $\Omega(n^{2})$ time \cite{DHZ00,CohenZ01,BaswanaK10,Kavitha12}, which raises the following question. 

\begin{question}
Can we construct a 2-approximate distance oracle in subquadratic time?
\end{question}

\subsection{Our Results}

Throughout we write $n^{\omega(r)}$ for the time for multiplying an $n\times n^r$ matrix by an $n^r \times n$ matrix (see \autoref{sec:prelim}).  Rectangular matrix multiplication is an active research field, with the bounds on $\omega(r)$ being improved in recent years~\cite{gall2023faster, GallU18, le2012faster}. We provide how the 
recent work by Vassilevska Williams, Xu, Xu, and Zhou~\cite{VassilevskaWXXZ23} affects our running times in \autoref{app:newRMM}. For the rest of this paper, we use \cite{GallU18}, the last published paper on the topic, for the sake of replicability. We balance the terms using \cite{balancer}.

\paragraph{Unweighted 2-Approximate APSP.}

Our first contribution is a faster algorithm for 2-approximate APSP in unweighted undirected graphs.\footnote{All our randomized algorithms are always correct; the randomness only affects the running time.}

\begin{restatable}{theorem}{ThmTwoApx}\label{thm:two_apx}
There exists a randomized algorithm that, given an unweighted, undirected graph $G=(V,E)$, computes $2$-approximate APSP.  With high probability, the algorithm takes $\tilde O(n^{2.5-r}+n^{\omega(r)})$ time, for any $r\in [0,1]$. Using known upper bounds for rectangular matrix multiplication, this is $O(n^{2.032})$ time. 
\end{restatable}

Our running time improves over the previous best running time of $\tilde{O}(\min \{m\sqrt{n}+n^2, n^{2.25} \})$ \cite{BaswanaK10,Kavitha12,Roditty23} as long as $m\geq n^{1.532}$, and gets closer to an $O(n^2)$ running time. In particular, if we have a rectangular matrix multiplication bound  of $\omega(0.5)=2$, then we obtain a $\tilde O(n^2)$ time algorithm. Currently, we know $\omega(0.5)<2.043$ \cite{VassilevskaWXXZ23} -- a very recent improvement on \cite{GallU18}: $\omega(0.5)<2.044183$.


While the fastest version of our algorithm exploits fast rectangular matrix multiplication, our approach also leads to an alternative simple combinatorial 2-approximation algorithm in $\tilde{O}(n^{2.25})$ time, matching the very recent result by Roditty~\cite{Roditty23}. See \autoref{sec:combinatotial} for the details.


\begin{table}[h]
\begin{center}
\begin{tabular}{|c|c|c|c|} 
\hline
 \multicolumn{3}{|c|}{Algorithms for Unweighted APSP} \\
 \hline
    Reference & Time & Approx. \\ \hline
    \cite{AingworthCIM99} & $n^{2.5}$ &  $+2$ \\  \hline
    \cite{DHZ00} & $ n^{7/3}$ &  $ +2$ \\  \hline
    \cite{DengKRWZ22} & $n^{2.287}$ &  $ +2$ \\  \hline
    \cite{durr2023improved} & $n^{2.260}$ &  $ +2$ \\  \hline
    \cite{Roditty23} & $n^{2.25}$ &  $ 2$ \\  \hline
    \textbf{This work} & $n^{2.032}$ &  $ 2$ \\  \hline

\end{tabular}
\quad
\begin{tabular}{|c|c|c|c|}
\hline

  \multicolumn{3}{|c|}{Algorithms for Weighted APSP} \\
 \hline
 Reference & Time & Approx. \\ \hline
  \cite{CohenZ01} & $n^{3/2}m^{1/2}$ &  $2$ \\  \hline
    \cite{BaswanaK10} & $m\sqrt{n}+n^2$ &  $2$ \\  \hline
  \cite{Kavitha12} & $n^{2.25}$ &  $2+\epsilon$ \\  \hline
  \textbf{This work} & $n^{2.214}$ &  $2+\epsilon$ \\  \hline
   \textbf{This work} & $mn^{2/3}$ &  $2$ \\  \hline
\end{tabular}
\end{center}
\caption{Algorithms for 2-Approximate APSP. The last result in the right table is a 2-approximate distance oracle.}
\label{table_apsp}
\end{table}

\paragraph{Weighted 2-Approximate APSP.} For weighted undirected graphs we show the following.

\begin{restatable}{theorem}{ThmTwoApxWeighted}\label{thm:two_apx_wghted}
There exists a randomized algorithm that, given an undirected graph $G=(V,E)$ with non-negative integer weights bounded by $W$, computes $(2+\epsilon)$-approximate APSP. With high probability the algorithm takes $O(n^{3-r} +n^{\omega(r)}(1/\epsilon)^{O(1)}\log W)$  time, for any $r\in [0,1]$. Using known upper bounds for rectangular matrix multiplication, this is $O(n^{2.214}(1/\epsilon)^{O(1)}\log W)$  time.
\end{restatable}

We remark that with standard scaling techniques the result generalizes to non-integer weights. Hence we state all our weighted results only for integers. 

To the best of our knowledge, this is the first improvement for weighted $(2+\epsilon)$-approximate APSP since the work of Kavitha~\cite{Kavitha12} that showed an algorithm with $\tilde{O}(n^{2.25})$ running time, that also exploited fast matrix multiplication techniques. We note that our algorithm and the algorithm of \cite{Kavitha12} give a $(2+\epsilon)$-approximation, and currently the fastest algorithms that give a 2-approximation for weighted APSP take $\tilde{O}( \min\{ n^{\omega}, m\sqrt{n}+n^2 \})$ time \cite{Zwick02,BaswanaK10}. The fastest combinatorial algorithm for the problem is the $\tilde{O}(m\sqrt{n}+n^2)$ time algorithm of Baswana and Kavitha~\cite{BaswanaK10}. This algorithm takes $\tilde{O}(n^{2.5})$ time in dense graphs, but it is faster for sparser graphs.

Our approach can also be combined with the approach of \cite{BaswanaK10} to obtain faster algorithms for sparser graphs, giving the following. See \autoref{table_weighted} for some specific choices for the value of $m$.

\begin{restatable}{theorem}{ThmTwoApxWeightedPar}\label{thm:two_apx_wghted_par}
    There exists a randomized algorithm that, given an undirected graph $G=(V,E)$ with non-negative, integer weights bounded by $W$ and a parameter $r\in [0,1]$, computes $(2+\epsilon)$-approximate APSP. With high probability, the algorithm runs in $\tilde O(mn^{1-r}+n^{\omega(r)}(1/\epsilon)^{O(1)}\log W)$ time.
\end{restatable}

In particular, our algorithm is faster than \cite{BaswanaK10} for $m\geq  n^{\omega(0.5)-0.5}$. With the current value of $\omega(0.5)$ our algorithm is faster for $m\geq n^{1.545}$~\cite{GallU18}. 

\paragraph{Weighted 2-Approximate Distance Oracle.}
The results mentioned above are fast for graphs that are relatively dense. For sparser graphs we show a simple combinatorial algorithm that gives the following.
\begin{restatable}{theorem}{thmstatic}\label{thm:static_2}
    There exists a combinatorial algorithm that, given a weighted graph $G=(V,E)$, constructs a distance oracle that answers $2$-approximate distance queries in constant time, and uses $\tilde O(mn^{2/3})$ space with preprocessing time $\tilde O(mn^{2/3})$. 
\end{restatable}

For sparse graphs ($m = o(n^{4/3})$), we have \emph{subquadratic} running time and space.  To the best of our knowledge this is the first 2-approximate distance oracle with subquadratic construction time for sparse graphs. 
Moreover, when $m=O(n)$ our bounds match the conditional lower bound of $\tilde \Omega(m^{5/3})$ preprocessing time for $2$-approximations, conditional on the set intersection conjecture~\cite{PatrascuRT12} or the 3-SUM conjecture~\cite{aboud2022stronger}, where \cite{PatrascuRT12} shows the stronger $\tilde \Omega(m^{5/3})$ \emph{space} lower bound. 
Moreover, for any stretch strictly below $2$, there is a $\tilde \Omega(n^2)$ lower bound (conditional on the set intersection conjecture)~\cite{PatrascuR14}. In \cite{PatrascuR14} they also show a polynomial time construction of a 2-approximate distance oracle for weighted graphs of size $\tilde{O}(n^{4/3} m^{1/3})$, which is subquadratic space for sparse graphs.
The proof of \autoref{thm:static_2} can be found in \autoref{sc:static_distance_oracles}. 



We can also obtain a $(2,W_{u,v})$-approximate distance oracle with a better construction time than our $2$-approximate construction. In particular, we can obtain a \emph{subquadratic} preprocessing time of $\tilde O(nm^{2/3})$ when $m \leq n^{3/2}$, see \autoref{app:static2W} for full details. This result is also implied by an algorithm of Baswana, Goyal, and Sen~\cite{BaswanaGS09}. They focus on $(2,1)$-distance oracles for unweighted graphs using the same algorithm. We observe that 1) the same algorithm can lead to subquadratic preprocessing time (in sparse graphs), and 2) this leads to a 
 $(2,W_{u,v})$-approximation for weighted graphs. We make these observations explicit for a wider range of parameter settings and for completeness give an analysis in \autoref{app:static2W}. 

\begin{table}[h]
\begin{center}
\begin{tabular}{|c|c|c|c|} 
\hline
    $m$ & This work & Previous results \\ \hline
    $n^{1.0}$ & $n^{1.667}$ &  $n^{2}$ \cite{BaswanaK10} \\  \hline
    $n^{1.1}$ & $n^{1.767}$ &  $n^{2}$ \cite{BaswanaK10} \\  \hline
    $n^{1.2}$ & $n^{1.867}$ &  $n^{2}$ \cite{BaswanaK10} \\  \hline
    $n^{1.3}$ & $n^{1.967}$ &  $n^{2}$ \cite{BaswanaK10} \\  \hline
    $n^{1.4}$ & $n^{2.009}$ &  $n^{2}$ \cite{BaswanaK10} \\  \hline
    $n^{1.5}$ & $n^{2.032}$ &  $n^{2}$ \cite{BaswanaK10} \\  \hline

\end{tabular}
\quad
\begin{tabular}{|c|c|c|c|}
\hline

    $m$ & This work & Previous results \\ \hline
    $n^{1.6}$ & $n^{2.062}$ &  $n^{2.1}$ \cite{BaswanaK10} \\  \hline
    $n^{1.7}$ & $n^{2.097}$ &  $n^{2.2}$ \cite{BaswanaK10} \\  \hline
    $n^{1.8}$ & $n^{2.134}$ &  $n^{2.25}$ \cite{Kavitha12}\\  \hline
    $n^{1.9}$ & $n^{2.173}$ &  $n^{2.25}$ \cite{Kavitha12} \\  \hline
    $n^{2.0}$ & $n^{2.214}$ &  $n^{2.25}$ \cite{Kavitha12}\\  \hline
\end{tabular}
\end{center}
\caption{Comparison of the novel running time bounds of our approximation algorithms for weighted graphs and prior work, for $1/\epsilon=n^{o(1)}$. For sparser graphs ($m \leq n^{1.3}$), we use our 2-approximate distance oracle (\autoref{thm:static_2}), and for denser graphs ($m \geq n^{1.4}$) we use our $(2+\epsilon)$-approximate APSP (\autoref{thm:two_apx_wghted_par}). \autoref{thm:static_2} gives the fastest running time when $m \leq n^{4/3}$, and \autoref{thm:two_apx_wghted_par} gives the fastest running time when $m\geq n^{1.545}$. For $n^{1.334} < m < n^{1.545}$, we do not improve on \cite{BaswanaK10}.}
\label{table_weighted}
\end{table}

\paragraph{Near-Additive APSP.}

We also study algorithms that give a near-additive approximation in unweighted undirected graphs. We show the following.

\begin{restatable}{theorem}{ThmNearAdditive}\label{thm:near_additive}
There exists a deterministic algorithm that, given an unweighted, undirected graph $G=(V,E)$ and an even integer $k\geq 2$, computes $(1+\epsilon,k)$-approximate APSP with running time $\tilde O\left(n^{2+(1-r)\tfrac{2}{k+2}}+n^{\omega(r)}(1/\epsilon)\right)$, for any choice of $r\in [0,1]$.  
\end{restatable}

As an important special case, we get a $(1+\epsilon,2)$-approximation in $O(n^{2.152})$ time when $1/\epsilon=n^{o(1)}$. 
While the running time is higher compared to our multiplicative 2-approximation, it improves over previous additive or near-additive approximation algorithms. As mentioned above, currently the fastest algorithm for a +2-approximation takes $O(n^{2.260})$ time \cite{durr2023improved}, and it exploits fast matrix multiplication. The fastest combinatorial +2-approximation algorithm takes  $\tilde{O}(\min\{n^{3/2}m^{1/2}, n^{7/3} \})$ time \cite{DHZ00}. If we consider near-additive approximation algorithms,  Berman and Kasiviswanathan showed a $(1+\epsilon,2)$-approximation in $n^{2.24+o(1)}$ time based on fast matrix multiplication \cite{berman2007faster}.

For $k \to \log n$, our running time goes to $\tilde O(n^2)$. In particular, when $1/\epsilon=n^{o(1)}$ we obtain 
a $(1+\epsilon,4)$-approximation in $O(n^{2.119})$ time, a $(1+\epsilon,6)$-approximation in $O(n^{2.098})$ time, and a $(1+\epsilon,8)$-approximation in $O(n^{2.084})$ time. The state of the art for these approximations is $O(n^{2.2})$, $O(n^{2.125})$, and $O(n^{2.091})$~\cite{DHZ00} respectively, using the purely additive algorithm by Dor, Halperin, and Zwick \cite{DHZ00} that gives a deterministic $+k$-approximation for APSP in $\tilde O(\min\{ n^{2-\tfrac{2}{k+2}}m^{\tfrac{2}{k+2}},n^{2+\tfrac{2}{3k-2}}\})$ time, for even $k>2$. See also \autoref{table_near_additive} in \autoref{sec:near_additive} for comparison of our results and prior work. 
We remark that the algorithm of \cite{DHZ00} becomes faster than ours once the additive term is at least 10. 
Our work shows that for smaller additive terms, a $(1+\epsilon,\beta)$-approximation can be obtained faster than a purely additive $+\beta$-approximation. To the best of our knowledge, previously such results were known only for the special case of $\beta=2$ \cite{berman2007faster}. 

We remark that except \cite{berman2007faster} previous algorithms for near-additive approximations \cite{cohen2000polylog,elkin2005computing,DBLP:conf/swat/ElkinGN22} had larger additive terms compared to the ones we study here.
In particular, Elkin, Gitlitz and Neiman \cite{DBLP:conf/swat/ElkinGN22} showed an algorithm that computes a $(1+\epsilon,\beta)$-approximation in $\tilde{O}(mn^{1/k}+n^{2+1/{(2^k-1)}})$ time, where $\beta=O(k/\epsilon)^{O(k)}$. While $\beta$ is a constant when $\epsilon$ and $k$ are constants, it is larger compared to the constants we consider here. 



\paragraph{Independent work} 
Independently, Saha and Ye~\cite{SahaY23} achieved similar results. They also obtain $2$-approximate APSP for unweighted graphs in $O(n^{2.032})$ time ({\autoref{thm:two_apx}}), moreover they give a deterministic algorithm for this problem in $O(n^{2.073})$ time. In the (near) additive regime (\autoref{thm:near_additive}), they both give faster running times and do not incur the multiplicative $(1+\epsilon)$ error. They include additional results for additive approximations in weighted graphs, but do not have our results for $2$-approximate APSP on weighted graphs (\autoref{thm:two_apx_wghted_par} and \autoref{thm:static_2}).

\subsection{Technical Overview}

\subsubsection*{Unweighted 2-Approximate APSP}

We start by describing our 2-approximation algorithm for APSP in unweighted undirected graphs running in $O(n^{2.032})$ time.
At a high-level, we divide all the shortest paths in the graph into 2 types: the \emph{sparse} paths and the \emph{dense} paths. A path is sparse if the degrees of all vertices in the path are at most $\sqrt{n}$, and it is dense otherwise. 

\paragraph{Dealing with sparse paths.} In order to compute the distances between all pairs of vertices $u,v$ where the shortest path between $u$ and $v$ is sparse we use the following approach. All the sparse paths are contained in a subgraph $G'=(V,E')$ of the input graph $G=(V,E)$ obtained by taking all edges adjacent to vertices of degree $\leq \sqrt{n}$. Note  that this graph has only $O(n^{3/2})$ edges. To estimate the distances, we run the algorithm of Baswana and Kavitha \cite{BaswanaK10} on the graph $G'$, and exploit the fact that this algorithm is efficient for sparse graphs. This gives a 2-approximation for the distances in the sparse graph $G'$ in $\tilde{O}(|E'|\sqrt{n}+n^2)=\tilde{O}(n^2)$ time.

\paragraph{Dealing with dense paths in $\tilde{O}(n^{2.5})$ time.} We are left with dense shortest paths, i.e. paths that have at least one vertex with degree larger than $\sqrt{n}$. As a warm-up, we start by describing a simple algorithm that obtains +2-approximations for the lengths of all these paths in $\tilde{O}(n^{2.5})$ time, following the classic algorithm of Aingworth, Chekuri, Indyk and Motwani \cite{AingworthCIM99}, and later we explain how we get a faster algorithm. We denote by $H$ the high-degree vertices that are vertices with degree larger than $\sqrt{n}$. We start by computing a \emph{hitting set} $S$, a small set of vertices such that each vertex in $H$ has a neighbor in $S$. It is easy to find a hitting set of size $\tilde{O}(\sqrt{n})$, for example by adding each vertex to the set with probability $\tilde{O}(1/\sqrt{n})$. Since each high-degree vertex has degree at least $\sqrt{n}$, with high probability all high-degree vertices will have neighbors in $S$. The algorithm then proceeds as follows.

\begin{enumerate}
\item We compute distances from $S$ to all other vertices.\label{step_mssp}
\item We set $\delta(u,v)=\min_{a \in S} \{ d(u,a)+d(a,v) \}.$\label{step_estimate}
\end{enumerate}

If the shortest path between $u$ and $v$ is dense, then after Step \ref{step_estimate} the value $\delta(u,v)$ is a +2-approximation for $d(u,v)$. The reason is that there exists a high-degree vertex $x$ in the shortest $u-v$ path, and $x$ has a neighbor $a \in S$. Then $d(u,a)+d(a,v) \leq d(u,x)+d(x,a)+d(a,x)+d(x,v)=d(u,v)+2$. See \autoref{stretch_pic} for illustration.
Hence by computing distances from $S$ to all other vertices, we can get an additive +2-approximation for all dense paths.

\setlength{\intextsep}{4pt}
\begin{figure}[h]
\centering
\setlength{\abovecaptionskip}{0pt}
\setlength{\belowcaptionskip}{8pt}
\includegraphics[scale=0.6]{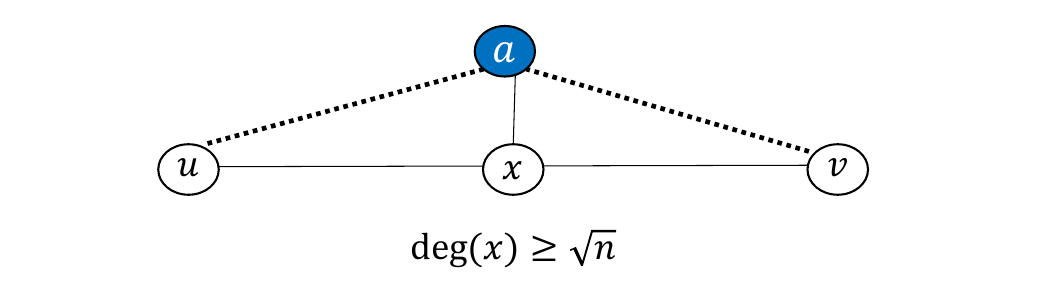}
 \caption{Illustration of the stretch analysis.}
\label{stretch_pic}
\end{figure}

The running time is $\tilde{O}(n^{2.5})$. First, computing distances from all vertices in $S$ takes $O(|S|m)=\tilde{O}(n^{2.5})$ time since $|S|=\tilde{O}(\sqrt{n}), m=O(n^2)$ and computing the distances from one vertex in $S$ takes $O(m)$ time by computing a BFS tree. Second, in Step \ref{step_estimate}, for each one of the $n^2$ pairs of vertices we compute distances through all possible vertices in $S$ which takes $\tilde{O}(n^{2.5})$ time.

\paragraph{Implementing Step \ref{step_mssp} faster.} Our goal is to implement the above approach faster. First, note that Step \ref{step_mssp} takes $O(|S|m)$ time. To obtain a faster algorithm, our goal is to run this step on a graph where $O(|S|m)=\tilde{O}(n^2)$. To do so, we divide the dense paths to $O(\log{n})$ different levels. We say that a path is $2^i$-dense if the maximum degree of a vertex in the path is between $[2^i,2^{i+1}]$. Since $2^i$-dense paths have a vertex with degree at least $2^i$, we can find a hitting set $S_i$ of size $\tilde{O}(n/2^i)$ such that each $2^i$-dense path will have a neighbor in $S_i$. Our goal is to repeat the above algorithm but on a sparser graph. Let $G_i=(V,E_i)$ be a subgraph 
of $G$ that has all edges adjacent to vertices of degree at most $2^{i+1}$. To deal with $2^i$-dense paths we work as follows.

\begin{enumerate}
\item We compute distances from $S_i$ to all other vertices in the graph $G_i$.\label{step_mssp_i}
\item We set $\delta(u,v)=\min_{a \in S_i} \{ \delta(u,a)+\delta(a,v) \}.$\label{step_estimate_i}
\end{enumerate}

In Step \ref{step_estimate_i}, the distance estimates $\delta(u,a),\delta(a,v)$ are the estimates computed in Step \ref{step_mssp_i}.
By definition, all the vertices of the $2^i$-dense paths and their edges to their neighbors are included in the graph $G_i$, so it is enough to compute distances in this graph in order to obtain +2-approximation for the lengths of $2^i$-dense paths. Since we worked on a sparser graph the running time for computing the BFS trees in level $i$ is now $O(|S_i||E_i|)=\tilde{O}((n/2^i) \cdot n  2^{i+1})=\tilde{O}(n^2)$. Summing up over all $O(\log{n})$ levels gives $\tilde{O}(n^2)$ running time for computing the BFS trees. After this step, we are guaranteed that if the shortest path between $u$ and $v$ is $2^i$-dense, then there is a vertex $a \in S_i$ such that $\delta(u,a)+\delta(a,v)\leq d(u,v)+2$, where $\delta(u,a),\delta(a,v)$ are the distances computed from $a$ in $G_i$.

\paragraph{Implementing Step \ref{step_estimate} faster.} By now we have computed all the relevant distances from the sets $S_i$, there is a remaining challenge. In order to estimate the distance of each pair $u$ and $v$ we need distance estimates going through all possible vertices $a \in S_i$ (Step \ref{step_estimate}), which takes $\tilde{O}(n^{2.5})$ time for all pairs, as in the worst case $|S_i|=\tilde{O}(\sqrt{n})$. 
To implement this step faster, we exploit fast matrix multiplication. Note that Step \ref{step_estimate} is essentially equivalent to matrix multiplication when the operations are minimum and plus, such multiplications are called \emph{distance products}. While it is well-known that APSP can be computed using matrix multiplication, usually it requires multiplication of square matrices which takes $\tilde{O}(n^{\omega})$ time. The trick in our case is that since we only want to compute distances through a small set $S_i$, we can use fast \emph{rectangular} matrix multiplication, as first exploited by Zwick~\cite{Zwick02}. Since $|S_i|=\tilde{O}(\sqrt{n})$ we only need to multiply matrices with dimensions $n \times n^{0.5}$ and $n^{0.5} \times n$, which can be done in just $O(n^{2.045})$ time using the rectangular matrix multiplication algorithms by Le Gall and Urrutia \cite{GallU18}. 
Fast matrix multiplication algorithms do not directly apply to distance products, however there is a well-known reduction that shows that we can get $(1+\epsilon)$-approximation for distance products in the same time \cite{Zwick02}, see \autoref{sec:prelim} for the details.

\paragraph{Conclusion.} Using the ideas described we can get a $(2+\epsilon)$-approximation in $O(n^{2.045})$ time. To remove the $\epsilon$ term in the stretch, we exploit the fact that in unweighted graphs we can get an additive $+\log{n}$-approximation in $\tilde{O}(n^2)$ time \cite{DHZ00}. Note that for pairs of vertices at distance larger than $\log{n}$, an additive $+\log{n}$-approximation is already a multiplicative 2-approximation. Hence we can focus our attention on pairs of vertices at distance at most $\log{n}$ from each other. We show that this allows us to turn any $(2+\epsilon)$-approximation for unweighted graphs to a 2-approximation, as long as the algorithm depends polynomially on $1/\epsilon$ (see \autoref{sc:2eps_to_2} for the details).

 Moreover, we can improve the running time to $O(n^{2.032})$ by a better balancing of our two approaches, for dealing with sparse and dense paths. In particular, in our discussion so far we defined sparse paths to be ones where the maximum degree is at most $\sqrt{n}$, which led to hitting sets of size $\tilde{O}(\sqrt{n})$. To obtain a faster algorithm we want to have a smaller hitting set of size $O(n^r)$ for an appropriate choice or $r$, and then the sparse paths are paths where the maximum degree is at most $O(n^{1-r})$. With these parameters, computing distances in the sparse graph using \cite{BaswanaK10} takes $\tilde{O}(n^{2.5-r})$ time, while dealing with dense paths takes $\tilde{O}(n^{\omega(r)})$ time. Balancing these two terms gives an $O(n^{2.032})$ time algorithm. Full details appear in \autoref{sec:unweighted_APSP}. 

 We remark that we can also implement Step \ref{step_mssp} using the $(1+\epsilon)$-approximate multi-source shortest paths algorithm by Elkin and Neiman \cite{ElkinN22} that is based on fast rectangular matrix multiplication. This allows computing $(1+\epsilon)$-approximate distances from $\tilde O(n^{r})$ sources in $\tilde{O}(n^{\omega(r)})$ time, which leads to the same overall running time. If we do so we can have only one set $S$ as in the original description of the algorithm. In our paper we implement Step~\ref{step_mssp} using the combinatorial algorithm discussed above that takes $\tilde{O}(n^2)$ time, which shows that currently the bottleneck in the algorithm is Step~\ref{step_estimate}, and other parts of the algorithm can be computed in $\tilde{O}(n^2)$ time. 

\paragraph{A combinatorial algorithm.} Our approach also leads to a simple combinatorial 2-approximation algorithm that takes $\tilde{O}(n^{2.25})$ time. To do so, we just change the threshold of sparse and dense paths. We say that a path is sparse if all the vertices in the path have degree at most $n^{3/4}$, and it is dense otherwise. Computing 2-approximations for sparse paths will now take $\tilde{O}(n^{2.25})$ time by \cite{BaswanaK10}. In the dense case, since the dense paths now have a vertex of degree at least $n^{3/4}$, we can construct a smaller hitting set $S$ of size $\tilde{O}(n^{1/4})$, and then we can implement Steps \ref{step_mssp} and \ref{step_estimate} directly via a combinatorial algorithm in $\tilde{O}(n^{2.25})$ time (by the same approach described above but replacing the size of $S$ with $\tilde{O}(n^{1/4})$). See \autoref{sec:combinatotial} for the details.

We remark that Roditty \cite{Roditty23} recently obtained the same result (a combinatorial 2-approximation in $\tilde{O}(n^{2.25})$ time) via a different approach. At a high-level, his approach is based on a detailed case analysis of the $+4$-approximation algorithm of Dor, Halperin, and Zwick~\cite{DHZ00}, showing that for close-by pairs a better approximation can be obtained.

\paragraph{Near-additive approximations.} We can use the same approach also in order to obtain near-additive approximations. Note that in the dense case, our algorithm actually computed a $(1+\epsilon,2)$-approximation, where the $(1+\epsilon)$ term comes from using fast matrix multiplication. Hence, if on the sparse graph we run the $+k$-additive approximation algorithm by Dor, Halperin, and Zwick~\cite{DHZ00} instead of the multiplicative 2-approximation algorithm of Baswana and Kavitha~\cite{BaswanaK10}, we can get $(1+\epsilon,k)$-approximations for all the distances. See \autoref{sec:near_additive} for the details. 

\subsubsection*{Weighted 2-Approximate APSP}
The techniques above do not generalize well to the weighted setting. For weighted graphs, we use a different approach, based on the set-up of \emph{bunches} and \emph{clusters} as introduced by Thorup and Zwick~\cite{TZ2005}. For a parameter $p\in [\tfrac{1}{n},1]$ we sample each vertex with probability $p$, and if sampled we add it to a set $S$. With high probability, we have $|S|=\tilde O(pn)$. Now, for each vertex $u\in V$, we define the pivot of $u$ to be the closest vertex in $S$ to $u$, i.e., $p(u)$ is an arbitrary vertex in the set $ \{v\in S : d(u,v)=d(u,S) \}$, and we define the bunch of $u$ by $B(u):= \{v\in V : d(u,v) < d(u,p(u))\} \cup \{p(u)\}$. For a vertex $v\in V$, we define the cluster of $u$ as the inverse bunch: $C(v) := \{u\in V : v\in B(u)\}$. Thorup and Zwick~\cite{TZ2005} show how to compute pivots, bunches, clusters, and distances $d(u,v)$ for all $u\in V,v\in B(u)$ in time $\tilde O(\tfrac{m}{p})$. Moreover, in follow up work~\cite{TZ01} they show that with different techniques, that is, a more involved construction of $S$, we can have that both bunches and clusters are bounded by $\tilde O(\tfrac{1}{p})$ with high probability. The running time remains $\tilde O(\tfrac{m}{p})$.

\paragraph{Warm-up: $3$-approximate APSP.} 
To see how we use this to compute approximate shortest paths, we consider a pair of vertices $u,v\in V$. If either $v\in B(u)$ or $u\in B(v)$, we have the exact distance, so assume this is not the case. In particular if $v\notin B(u)$, then $d(u,p(u))\leq d(u,v)$. We compute shortest paths from $S$ to $V$, in particular obtaining $d(p(u),v)$. With this additional information we get a $3$-approximation almost directly~\cite{TZ2005}:
\begin{align*}
    d(u,p(u))+d(p(u),v) \leq d(u,p(u)) +d(p(u),u)+d(u,v) \leq 3d(u,v),
\end{align*}
where the first inequality holds by the triangle inequality. 

\paragraph{$2$-approximate APSP.}
By a closer inspection, we can improve this analysis to a $2$-approximation for each pair of vertices, whose shortest path interacts in a particular way with the bunches. To be precise, let $u,v\in V$ be a pair of vertices and let $\pi$ be a shortest path between them. Suppose it contains a vertex $x$, such that $x\notin B(u)$ and $x\notin B(v)$, see the right case in \autoref{fig:adjacent}. That means that $d(u,p(u))\leq d(u,x)$ and $d(v,p(v))\leq d(x,v)$, so at least one of the two is at most $d(u,v)/2$. Without loss of generality, say that $d(u,p(u))\leq d(u,v)/2$. Hence, after computing shortest paths from $S$, we can obtain a $2$-approximation as follows: $d(u,p(u))+d(p(u),v)\leq d(u,p(u))+d(u,p(u))+d(u,v)\leq 2d(u,v)$, where the first inequality holds by the triangle inequality. 

\begin{figure}
    \centering
    \includegraphics[width=\textwidth]{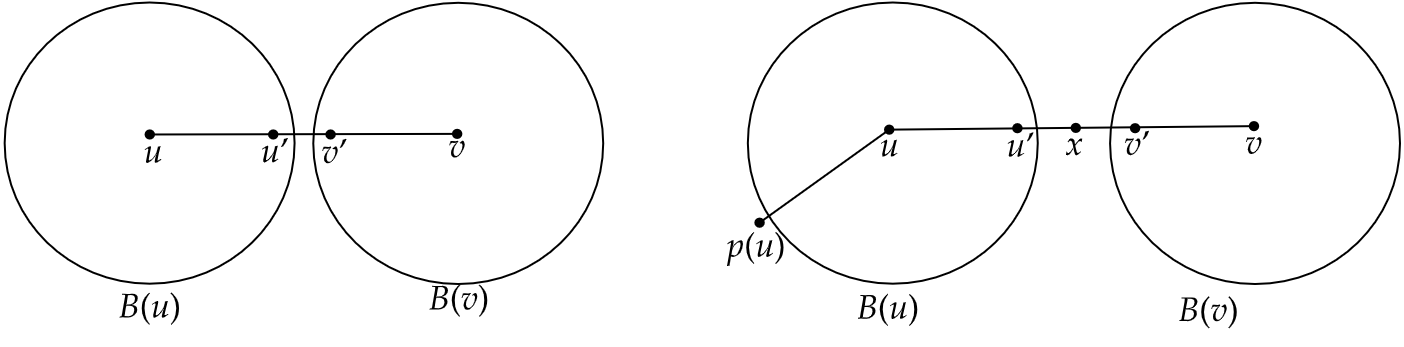}
    \caption{Two different possible interactions between the shortest path between $u$ and $v$, and the bunches of $u$ and $v$.}
    \label{fig:adjacent}
\end{figure}

It remains to give an algorithm that guarantees a $2$-approximation for the case that for every vertex~$x$ on the shortest path~$\pi$ we have $x\in B(u)$ or $x\in B(v)$, the left case in \autoref{fig:adjacent}. Note that this case also contains the special case that the bunches overlap. We will refer to it as the `adjacent case', since the two bunches $B(u)$ and $B(v)$ are adjacent in the sense that they are connected by an edge. In previous work, this adjacent case is often a bottleneck in the running time, and dealt with in various ways. As seen above, by not distinguishing it at all, we obtain a $3$-approximation~\cite{TZ2005}. 
By only considering the cases where the bunches have at least one vertex in common, Baswana, Goyal, and Sen~\cite{BaswanaGS09} obtain a $(2,1)$-approximation. In \autoref{app:static2W} we generalize this result to $(2,W_{u,v})$-approximate APSP in weighted graphs, where $W_{u,v}$ is the maximum weight of an edge on a shortest $u-v$ path. For a $2$-approximation, Kavitha~\cite{Kavitha12} and Baswana and Kavitha~\cite{BaswanaK10} each use a multilevel approach, the latter of which we detail later. More recently, in distributed~\cite{censor2021fast,dory2022exponentially} and dynamic~\cite{DoryFNV22} $2$-approximate APSP algorithms, the adjacent case is computed explicitly. To be precise, they compute 
$$\delta_{\rm{adjacent}}(u,v) = \min\{ d(u,u')+w(u',v')+d(v',v) : \{u',v'\} \in E, u'\in B(u), v'\in B(v)\},$$ 
which in the case that $\pi \subseteq B(u)\cup B(v)$ returns the exact distance between $u$ and $v$. Using these algorithms directly does not lead to fast algorithms in our, centralized, setting, since they are tailored for their respective models. 
Our work is inspired by this approach, and gives two novel ways to compute $\delta_{\rm{adjacent}}(u,v)$; one for sparse graphs and one for dense graphs. 

\paragraph{Sparse case.}
Rather than fixing $u,v \in V$ and trying to compute $\delta_{\rm{adjacent}}(u,v)$, we fix an edge $\{u',v'\} \in E$, and see for which pairs $u,v \in V$ this contributes to $\delta_{\rm{adjacent}}(u,v)$. More precisely, by definition of bunches and clusters we have $u'\in B(u)\iff u\in C(u')$, so we can compute $\delta_{\rm{adjacent}}$ as follows: for all $\{u',v'\} \in E$, for all $u\in C(u')$, and for all $v\in C(v')$:
\begin{enumerate}
    \item Initialize $\delta_{\rm{adjacent}}(u,v) \leftarrow d(u,u')+w(u',v')+d(v',v)$ if no such entry exists. \label{step:adj_init_intro}
    \item Otherwise: $\delta_{\rm{adjacent}}(u,v)\leftarrow \min\{ \delta_{\rm{adjacent}}(u,v), d(u,u')+w(u',v')+d(v',v)\}$.\label{step:adj_update_intro}
\end{enumerate}
    We note that Step~\ref{step:adj_init_intro} and~\ref{step:adj_update_intro} can both be done in constant time, so computing $\delta_{\rm{adjacent}}$ takes time
\[ \sum_{\{u',v'\}\in E} \sum_{u\in C(u')} \sum_{v\in C(v')} O(1) = \sum_{\{u',v'\}\in E} O(|C(u')|\cdot |C(v')|) = \sum_{\{u',v'\}\in E} \tilde O(\tfrac{1}{p^2}), \]
where the last equality holds since clusters have size at most $\tilde O(\tfrac{1}{p})$. This means it takes $\tilde O(\tfrac{m}{p^2})$ time in total to compute $\delta_{\rm{adjacent}}$. 

Together with the running time for computing bunches and clusters, $\tilde O(\tfrac{m}{p})$, and the running time for computing shortest paths from $S$, $\tilde O(|S|m)=\tilde O(pnm)$ using Dijkstra, we obtain $\tilde O(\tfrac{m}{p^2}+pnm)$. For $p=n^{-1/3}$, we obtain running time $\tilde O(mn^{2/3})$.

We note that for each pair of vertices $u,v\in V$, we still have to take the minimum between the estimate through the pivot, $d(u,p(u))+d(p(u),v)$, and $\delta_{\rm{adjacent}}(u,v)$. Doing this explicitly would take $n^2$ time. Instead, we provide a distance oracle, and perform this minimum in constant time when the pair $u,v$ is queried.

\paragraph{Dense case.} 
We adapt our algorithm in two ways for the dense case. We use a different approach to compute $\delta_{\rm{adjacent}}$, and we compute shortest paths from the set of pivots $S$ differently. 

First, we show how to compute $\delta_{\rm{adjacent}}$ in $\tilde O(\tfrac{n^2}{p})$ time. For each node $u$, we run Dijkstra twice on a graph with $\tilde O(\tfrac{n}{p})$ edges, whose size comes from the fact that for each node the bunches have size $\tilde O(\tfrac{1}{p})$. On the first graph we obtain estimates from $u$ to $v'$ for every $v'$ that neighbors the bunch of $u$, i.e., $\exists u'\in B(u)$ such that $\{u', v'\} \in E$. And on the second graph we obtain estimates $\delta_{\rm{adjacent}}(u,v)$ for all $v\in V$. For more details see \autoref{sc:weighted_dense}.

Second, for computing shortest paths from $S$, we can do something (much) more efficient than computing multiple Dijkstra's by using recent results on approximate multi-source shortest paths (MSSP). Elkin and Neiman~\cite{ElkinN22} provide efficient $(1+\epsilon)$-MSSP results using rectangular matrix multiplication. For example, we can let the number of pivots be as big as $n^{0.8}$, while the running time stays below $O(n^{2.23})$. This means that the sizes of bunches drop dramatically to $\tilde O(n^{0.2})$, making the above very efficient. To be more precise, we need to balance the running time to compute $\delta_{\rm{adjacent}}$, $\tilde O(\tfrac{n^2}{p})$, with the running time to compute shortest paths from $S$, which has size $\tilde O(pn)$. If we use Dijkstra for the latter (for a graph with $m=n^2)$, we need to balance $\tfrac{n^2}{p}$ and $pn^3$, obtaining running time $\tilde O(n^{2.5})$ for $p=1/\sqrt{n}$.

To see how this trade-off improves using~\cite{ElkinN22}, we denote $p=n^{r-1}$. Now we have $|S|=\tilde O(pn)=\tilde O(n^r)$, and we can compute $(1+\epsilon)$-approximate shortest paths from $S$ in $\tilde O(n^{\omega(r)})$ time. 
We obtain total running time $\tilde O(\tfrac{n^2}{p}+n^{\omega(r)})=O(n^{2.214})$, using \cite{GallU18} for an upper bound on~$\omega(r)$.

We note that the $(1+\epsilon)$-factor carries over to our stretch analysis, making it a $(2+\epsilon)$-approximation, see \autoref{lm:weighted_correctness}.

\paragraph{A Density Sensitive Algorithm.}
Next, we describe how we can generalize the ideas from the dense case to a wider density range. Our goal is to combine our approach from the dense case with the $2$-approximate APSP algorithm of Baswana and Kavitha~\cite{BaswanaK10}. Similar to the dense case, they create a set of pivots $S$ of size $\tilde O(pn)$. They compute shortest paths from $S$ to $V$ using Dijkstra in $\tilde O(pnm)$ time, and use this for an estimate through the pivot.

Baswana and Kavitha~\cite{BaswanaK10} do not consider the adjacent case explicitly, but consider $O(\log n)$ additional levels, each with a gradually growing set of pivots $S_i$. For each of those sets, they compute shortest paths in a sparser graph, where the distances to some essential vertices equal the distances in the original graph. They show that on at least one of the levels, a distance estimate through the pivot of that level gives a $2$-approximation. For details we refer to \autoref{sc:parBK}.

By computing shortest paths on the sparser graph, they avoid the expensive computation of shortest paths from $S_i$ to all of $V$. Instead, for each level, they require $\tilde O(m/p)$ time to construct the sparser graph, and $\tilde O(n^2)$ time to compute shortest paths from $S_i$. 
Combining this with the first step, their algorithm takes $\tilde{O} (pnm + m/p+n^2)$ time. Setting $p=1/\sqrt{n}$ balances the terms and gives $\tilde O(m\sqrt n+ n^2)$ time. 

We modify their algorithm in two ways. First of all, we remark that parameterizing the size of the (smallest) set of pivots $S$, by $|S|=\tilde O(pn)$ already gives us the following result, which allows us to get a better trade-off. 

\begin{restatable}{theorem}{ThmParBK}\label{thm:ParBK}
    There exists a randomized algorithm that, given an undirected graph with non-negative edge weights $G=(V,E)$ and parameters $p\in (\tfrac{1}{n},1]$, $\epsilon\geq 0$, computes $(2+\epsilon)$-approximate APSP. With high probability, the algorithm takes $\tilde O(n^2+m/p+T(\tilde O(pn)))$ time, where $T(s)$ is the time to compute $(1+\epsilon)$-MSSP from $s$ sources.
\end{restatable}

Secondly, we use the fast MSSP algorithm of Elkin and Neiman~\cite{ElkinN22} to compute the shortest paths from $S$ faster. If we write $p=n^{q-1}$, for some $r\in [0,1]$, then \cite{ElkinN22} gives $T(pn)=T(n^r)=\tilde O(m^{1+o(1)}+n^{\omega(r)})$. The total running time for $2$-approximate APSP is then $\tilde O(m n^{1-r}+n^{\omega(r)})$, for any $r\in [0,1]$. 
For $m=n^2$ this recovers our result, \autoref{thm:two_apx_wghted}, for dense graphs. See \autoref{thm:two_apx_wghted_par} and \autoref{table_weighted} for our results for graphs with different densities.

\subsection{Additional Related Work} 

\paragraph{Approximation between 2 and 3.} In addition to the 2-approximate APSP algorithms mentioned above, approximate APSP algorithms with approximations between 2 and 3 in weighted undirected graphs are studied in \cite{CohenZ01,BaswanaK10,Kavitha12,AkavR21}. In particular, \cite{CohenZ01,BaswanaK10} studied algorithms with approximation $7/3$, and \cite{Kavitha12} studied an algorithm with approximation $5/2$. These results were generalized by Akav and Roditty~\cite{AkavR21} who showed an algorithm with approximation  $ 2 + \tfrac{k-2}{k} $ and running time $ \tilde O (m^{2/k} n^{2 - 3/k} + n^2) $ for any $ k \geq 2 $.

\paragraph{Algorithms using fast matrix multiplication.}
Algorithms for fast rectangular matrix multiplication are studied in 
\cite{coppersmith1982rapid,lotti1983asymptotic,coppersmith1997rectangular,DBLP:journals/jc/HuangP98,ke2008fast,le2012faster,GallU18}, and have found numerous applications in algorithms (see e.g. \cite{Zwick02,roditty2011all, yuster2009efficient, yuster2004detecting, kaplan2007counting, kaplan2006colored, sankowski2010fast, williams2014finding,BackursRSWW21, van2019dynamic, van2022fast,bergamaschi2021new, williams2020monochromatic,DBLP:conf/icalp/GuR21}).

In the context of APSP, the state of the art algorithm by Zwick for computing APSP in directed graphs with bounds weights is based on rectangular matrix multiplication \cite{Zwick02}. In addition, Kavitha \cite{Kavitha12} used fast rectangular matrix multiplication as one of the ingredients in her $(2+\epsilon)$-approximation algorithm for weighted APSP. Additive +2-approximations for APSP based on fast matrix multiplication are studied in \cite{DengKRWZ22,durr2023improved}. The latter algorithms are based on Min-Plus product
of rectangular bounded difference matrices. The $(1+\epsilon,2)$-approximation of Berman and Kasiviswanathan \cite{berman2007faster} is also based on fast rectangular matrix multiplication.
In addition, Elkin and Neiman showed $(1+\epsilon)$-approximation for multi-source shortest paths based on rectangular matrix multiplication \cite{ElkinN22}.
For the problem of computing shortest paths for $S \times T$, for $|S|,|T|=O(n^{0.5})$, Dalirrooyfard, Jin, Vassilevska Williams, and Wein~\cite{DalirrooyfardJW22} provide a $2$-approximation in $\tilde O(m+n^{(1+\omega)/8})$ time for weighted graphs, by leveraging \emph{sparse} rectangular matrix multiplication.
Dynamic algorithms for shortest paths and spanners based on rectangular matrix multiplication are studied in \cite{van2019dynamic,van2022fast,bergamaschi2021new}, and an algorithm for approximating the diameter based on rectangular matrix multiplication is studied in \cite{BackursRSWW21}.

\subsection{Discussion} 

In this work we showed fast algorithms for 2-approximate APSP, many intriguing questions remain open.
First, our algorithm for 2-approximate APSP in unweighted undirected graphs takes $O(n^{2.032})$ time, and an interesting direction for future work is to try to obtain an $O(n^2)$ time algorithm for this problem.

Second, many of our algorithms are based on fast rectangular matrix multiplication, and it would be interesting to develop also fast combinatorial algorithms for these problems. Currently the fastest combinatorial algorithm for 2-approximate APSP in unweighted undirected graphs is the recent algorithm by \cite{Roditty23} that takes $\tilde{O}(n^{2.25})$ time. For weighted undirected graphs the fastest combinatorial 2-approximate APSP algorithm takes $\tilde{O}(m\sqrt{n}+n^2)$ time, which is $\tilde{O}(n^{2.5})$ for dense graphs, where we show a non-combinatorial $(2+\epsilon)$-approximation algorithm that takes $O(n^{2.214})$ time. Narrowing the gaps between the combinatorial and non-combinatorial algorithms, or proving conditional hardness results for combinatorial algorithms is an interesting direction for future research. We remark that such a gap exists also for the case of additive $+2$-approximations, where the fastest combinatorial algorithm takes $\tilde{O}(\min\{n^{3/2}m^{1/2}, n^{7/3} \})$ time \cite{DHZ00}, and the fastest non-combinatorial algorithm takes $O(n^{2.260})$  time \cite{durr2023improved}.

\section{Preliminaries} \label{sec:prelim}
Throughout, we use $n=|V|$ and $m=|E|$. When writing $\log$-factors, we round them to the closest integer.
We use $d_G(u,v)$ for the distance between $u$ and $v$ in $G$, where we omit $G$ if it is clear from the context. We write $\delta(u,v)$ for distance estimates between $u$ and $v$. 
We denote SSSP for the Single Source Shortest Path Problem, MSSP for the Multi Source Shortest Path Problem, and APSP for the All Pairs Shortest Path Problem. SSSP takes the source as part of the input, and MSSP takes the set of sources as part of the input. 
We say that an algorithm gives an $(\alpha,\beta)$-approximation for SSSP/MSSP/APSP if for any pair of vertices $u,v$ it returns an estimate $\delta(u,v)$ of the distance between $u$ and $v$ such that $d(u,v) \leq \delta(u,v) \leq \alpha \cdot d(u,v) + \beta$, where $d(u,v)$ is the distance between $u$ and $v$. If $\beta=0$, we get a purely multiplicative approximation that we refer to as $\alpha$-approximation. If $\alpha=1$, we get a purely additive approximation, and call it a $+\beta$-approximation. If $\alpha=1+\epsilon$, we refer to it as a near-additive approximation. 

All our randomized algorithms are always correct, and provide running time guarantees `with high probability', which means with probability $1-n^{-c}$ for any constant $c$. 

When we refer to the APSP problem, we are required to output all $n^2$ distances. If we drop this requirement, and only want to give a data structure subject to distance queries, we call it a \emph{distance oracle}. For distance oracles, there are three complexities to consider: preprocessing time, space, and query time. Both the preprocessing time and the required space can be less than $n^2$. In our algorithms, we provide constant query time. 

\paragraph{Rectangular Matrix Multiplication and Shortest Paths}
Let $A$ be an $n_1\times n_2$ matrix, and $B$ be an $n_2 \times n_3$ matrix, then we define the \emph{distance product}, also called \emph{min-plus product}, $A\star B$ by 
$$ (A\star B)_{ij} = \min_{1\leq k \leq n_2} \{A_{ik}+B_{kj}\},$$
for $1\leq i \leq n_1$ and $1\leq j \leq n_3$.

Moreover we say a matrix $A'\in \mathbb R^{n_1\times n_2}$ is a $(1+\epsilon)$-approximation of a matrix $A\in \mathbb R^{n_1\times n_2}$ if $A_{ij} \leq A'_{ij} \leq (1+\epsilon) A_{ij}$, for all $1\leq i \leq n_1$ and $1\leq j \leq n_2$. 

Distance product form a \emph{semiring}, i.e., a ring without guaranteed additive inverses. Although the results for fast matrix multiplication hold for \emph{rings}, it turns out they can be leveraged for distance products as well~\cite{AlonGM97}. In particular, we have the following result for approximate distance products.

\begin{theorem}[\cite{Zwick02}]\label{thm:RMM_Zwick}
    Let $W$ be a positive integer, and $\epsilon>0$ be a parameter. Let $A$ be an $n\times n^r$ matrix and $B$ be an $n^r \times n$ matrix, whose entries are all in $\{0,1, \dots, W\}\cup \{\infty\}$. Then there is an algorithm that computes a $(1+\epsilon)$-approximation to $A\star B$ in time $\tilde O(n^{\omega(r)}/\epsilon \log W)$.
\end{theorem}

Here $\omega(r)$ denotes the time constant for rectangular matrix multiplication, i.e., the constant such that in $n^{\omega(r)}$ time we can multiply an $n\times n^r$ with an $n^r \times n$ matrix. 
This algorithm is a deterministic reduction to rectangular matrix multiplication, for which the state of the art~\cite{GallU18} is also deterministic. Note that \cite{GallU18} does not provide a closed form for $\omega(r)$. Throughout, we use the tool of van den Brand~\cite{balancer} to balance $\omega(r)$ with other terms to obtain our numerical results.

Backurs, Roditty, Segal, Vassilevska Williams, and Wein~\cite{BackursRSWW21} leverage these results to obtain multi-source approximate shortest paths from $\sqrt{n}$ sources, given that the distances are short. 
Elkin and Neiman~\cite{ElkinN22} show how to exploit rectangular matrix multiplication to compute distances from an arbitrary set of sources $S$.
\begin{theorem}[\cite{ElkinN22}]\label{thm:EN_MSSP}
    There exists a deterministic algorithm that, given a parameter $\epsilon>0$, an undirected graph $G=(V,E)$ with integer weights bounded by $W$ and a set of sources $S$ of size $|S|=O(n^r)$, computes $(1+\epsilon)$-approximate distances for $S\times V$ in $\tilde O(m^{1+o(1)}+n^{\omega(r)}(1/\epsilon)^{O(1)}\log W)$ time. 
\end{theorem}

\subsection{From \texorpdfstring{$2+\epsilon$}{2+epsilon} to \texorpdfstring{$2$}{2}-Approximate APSP}\label{sc:2eps_to_2}
Dor, Halperin, and Zwick~\cite{DHZ00} provide a $+\log n$-approximate APSP for unweighted graphs in $\tilde O(n^2)$ time. Using this result, we can reduce the problem of an unweighted $2$-approximation to an unweighted $(2+\epsilon)$-approximation. 

\begin{lemma}\label{lm:2eps_to_2}
    Given an algorithm that computes $(2+\epsilon)$-approximate APSP on unweighted graphs in $\tau(n,m)\poly(1/\epsilon)$ time, we obtain an algorithm that computes $2$-approximate APSP in $\tilde O(n^2+\tau(n,m))$ time. The reduction is deterministic and combinatorial.
\end{lemma}
\begin{proof}
    Set $\epsilon=1/\log n$, and let $\delta(u,v)$ denote the output of the $(2+\epsilon)$-approximate APSP algorithm. Let $\delta'(u,v)$ denote the output of the $+\log n$ approximation of~\cite{DHZ00} (see also \autoref{thm:DHZ}). 
    We output $\hat{d}(u,v) = \min\{ \lfloor \delta(u,v)\rfloor, \delta'(u,v)\}$.
    Clearly this takes $\tilde O(n^2+\tau(n,m))$ time in total, so it remains to show that $\hat{d}(\cdot,\cdot)$ is a $2$-approximation. 

    First of all we notice that since distances in an unweighted graph are integers, we have that $d(u,v) \leq \delta(u,v)$ implies that $d(u,v) \leq \lfloor\delta(u,v)\rfloor$. Since also $d(u,v)\leq \delta'(u,v)$, we can conclude $d(u,v)\leq \hat{d}(u,v)$. 

    To show that $\hat{d}(u,v)\leq 2d(u,v)$, we distinguish the cases $d(u,v)<\log n$ and $d(u,v)\geq \log n$. If $d(u,v)<\log n$, then from $\delta(u,v) \leq (2+\epsilon)d(u,v)$ we obtain 
    \begin{align*}
    \hat{d}(u,v) &\leq \lfloor \delta(u,v)\rfloor 
    \leq \lfloor (2+\epsilon)d(u,v)\rfloor \\
    &= 2d(u,v)+ \lfloor\epsilon d(u,v) \rfloor 
    =2d(u,v)+ \lfloor \tfrac{d(u,v)}{\log n} \rfloor \\
    &\leq 2d(u,v).
    \end{align*}
    Now if $d(u,v)\geq \log n$, we have that $\hat{d}(u,v)\leq \delta'(u,v)\leq d(u,v)+\log n\leq 2d(u,v)$.
\end{proof}

\section{Approximate APSP Algorithms for Unweighted Graphs} \label{sec:unweighted_APSP}
We obtain our three unweighted APSP results, $2$-approximate (\autoref{thm:two_apx}), combinatorial $2$-approximate (\autoref{thm:two_apx_comb}), and $(1+\epsilon,k)$-approximate for even $k\geq 2$ (\autoref{thm:near_additive}), through a more general framework. In this section, we develop said framework, from which these theorems follow almost immediately. 

The goal of our framework is to split the graph into two cases: a sparse graph and a dense graph. On the sparse graph, we just run an existing approximate APSP algorithm that performs well on sparse graphs (denoted by Algorithm~$\mathcal A$ in \autoref{thm:sparse_APSP} below). For the dense graph, we use more ingenuity. We further split it into $\poly\log n$ density regimes, where in each regime the bottleneck is to compute APSP through a known set $S$, i.e., for each $u,v$ to find a shortest path of the form $(u,x,v)$ for some $x\in S$ (this task is done by Algorithm~$\mathcal B$). If we solve this problem exactly, we obtain a $+2$-approximation. If we solve it approximately, this carries over into the overall approximation factor. For example, if we solve it up to a multiplicative factor $(1+\epsilon)$, we obtain total approximation $(1+\epsilon, 2(1+\epsilon))$ for the dense graph. For our approximations, we only use algorithms $\mathcal B$ that are either exact or $(1+\epsilon)$-approximate. 
The formal statement of the framework is as follows. 

\begin{restatable}{theorem}{ThmSparseAPSP}\label{thm:sparse_APSP}
Let $\mathcal A$ be an algorithm that computes $(\text{mult}_A,\text{add}_A)$-approximate APSP on unweighted graphs with running time $\tau_A(n,m)$, and let $\mathcal B(S)$ be an algorithm that computes $(\text{mult}_B,\text{add}_B)$-approximate all-pairs shortest paths on weighted graphs, among the $u-v$ paths of the form $(u,x,v)$ for some $x$ in a given set $S$, with running time $\tau_B(n,|S|)$. 
Then there exists an algorithm that, given an unweighted, undirected graph $G=(V,E)$ and a parameter $r\in[0,1]$, computes approximate APSP with running time $\tau_A(n,n^{2-r})+\tilde O(\tau_B(n,\tilde O(n^r)))$, where for each pair of vertices we have either a $(\text{mult}_A,\text{add}_A)$ or a $ (\text{mult}_B,\text{add}_B+2\text{mult}_B)$ approximation. 

Besides possibly Algorithm~$\mathcal{A}$ and~$\mathcal{B}$, the procedure is deterministic and combinatorial.
\end{restatable}
Note that $\tau_A(n,n^{2-r}),\tau_B(n,\tilde O(n^r)) \geq n^2$ as they both need $n^2$ time to write their output.

In \autoref{sc:unweighted}, \autoref{sec:combinatotial}, and \autoref{sec:near_additive}, we obtain a $2$-approximation, a combinatorial $2$-approximation, and a near-additive approximation respectively, by using different algorithms for $\mathcal{A, B}$ and balancing the parameter $r$ accordingly. 

Next, we proceed by describing the algorithm satisfying \autoref{thm:sparse_APSP}, followed by a correctness proof and running time analysis. Pseudo-code can be found in \autoref{alg:sparse_APSP}. To desecribe the algrithm, we recall the notion of hitting sets.
A set $S$ is said to be a \emph{hitting set} for the vertices that have at least one neighbor in $S$. The following result provides a hitting set $S$ for the vertices with degree at least $s$. Such a set can easily be obtained by random sampling or with a deterministic algorithm~\cite{AingworthCIM99,DHZ00}.

\begin{lemma}[\cite{DHZ00}]
    There exists a deterministic algorithm $\Hit(G,s)$ that, given an undirected graph $G=(V,E)$ and a parameter $1\leq s \leq n$, computes a set $S\subseteq V$ of size $O(\tfrac{n}{s}\log n)$ such that all vertices of degree at least $s$ have at least one neighbor in $S$. The algorithm takes $O(m+n)$ time.
\end{lemma}

\paragraph{Algorithm.}
Given the parameter $r\in[0,1]$, we say a vertex is \emph{light} if it has at most $n^{1-r}$ incident edges. We look at the \emph{sparse graph}, where each vertex keeps at most $n^{1-r}$ edges to its neighbors. On this graph we run Algorithm $\mathcal A$. We show that if for two vertices the shortest path only consists of light vertices, then this provides a $(\text{mult}_A,\text{add}_A)$-approximation. 
We also run the following procedure, on the entire input graph. This part ensures that if the shortest path between two vertices contains at least one dense vertex, we obtain a $(\text{mult}_B,\text{add}_B+2)$-approximation. We look at $O(\log n)$ levels.  At level $i$, the goal is to get an approximation for a shortest path with maximum degree in $[2^i,2^{i+1}]$. We let $i$ be all the integer values between $(1-r)\log n$ and $\log n$. By abuse of notation, we write $i=(1-r)\log n,(1-r)\log n+1, \dots, \log n$. At level $i$ we do the following. 
\begin{enumerate}
    \item  Let $S_i \subseteq V$ be defined by $\Hit(G,2^i)$. We set $G_i=(V,E_i)$ to be the graph with $E_i := \{ \{u,v\} \in E : \text{deg}(u)\leq 2^{i+1}\text{ or }\text{deg}(v)\leq 2^{i+1}\}$.
    \item We compute multi-source shortest paths from $S_i$ on $G_i$, by running Dijkstra from each vertex in $S_i$. We store these results in the graph $G_i':= (V,E_i')$. So $E_i' = S_i\times V$, and $w_{E_i'}(a,v):=d_{G_i}(a,v)$.
    \item Now we want to compute shortest paths through $S_i$ on the graph induced by edges in step (2), hereto we call algorithm $\mathcal B(G_i',S_i)$.
\end{enumerate}

\begin{algorithm}[H] \SetAlgoLined 
Let $G'=(V,E')$, where $E':= \{ \{u,v\} \in E : \deg(u)\leq n^{1-r} \text{ or } \deg(v)\leq n^{1-r} \}$ \;
Let $\delta_A(\cdot,\cdot)$ be the result of Algorithm $\mathcal A(G')$ \;
\For{$i=(1-r)\log n, \dots, \log n$}{
Let $G_i=(V,E_i)$, where $E_i := \{ \{u,v\} \in E : \deg(u)\leq 2^{i+1} \text{ or } \deg(v)\leq 2^{i+1} \}$\;
$S_i := \Hit(G,2^i)$\;
\For{$a\in S_i$}{
Run Dijkstra from $a$ on $G_i$}
Let $G_i':= (V,E_i')$ be a weighted graph, where $E_i' = S_i\times V$, and $w_{E_i'}(a,v):=d_{G_i}(a,v)$ for $a\in S_i,v\in V$.\;
Let $\delta_i(\cdot,\cdot)$ be the result of Algorithm $\mathcal B(G_i',S_i)$\;
}
$\delta_B(u,v) := \min\{ \delta_i(u,v) : (1-r)\log n \leq i \leq \log n\}$, for $u,v\in V$.\;
$\delta(u,v) := \min\{ \delta_A(u,v),\delta_B(u,v)\}$, for $u,v\in V$.\;
\Return{$\delta(\cdot,\cdot)$}
\caption{Our algorithm to compute APSP using algorithms~$\mathcal{A}$ and~$\mathcal B$} 
\label{alg:sparse_APSP}
\end{algorithm}

\paragraph{Correctness.}
Given $u,v\in V$, we have to show that the returned approximation $\delta(u,v)$ is a $2$-approximation. 
We distinguish two (non-disjoint) cases:

\begin{enumerate}[a)]
    \item There exists a shortest path from $u$ to $v$ solely consisting of light vertices: vertices of degree at most $n^{1-r}$. In this case, we show we obtain a $(\text{mult}_A,\text{add}_A)$-approximation through $\delta_A$. 
    
    In this case, this shortest path is fully contained in the sparse graph, hence running algorithm $\mathcal A$ on it provides the given approximation $\delta_A$.
    
    \item There exists a shortest path from $u$ to $v$ containing at least one heavy vertex: a vertex with degree at least $n^{1-r}$. In this case, we show we obtain a $(\text{mult}_B,\text{add}_B+2\text{mult}_B)$-approximation through $\delta_B$.
    
    Let $i\in \{(1-r)\log n,(1-r)\log n+1, \dots \log n\}$ be the maximal index such that there exists a vertex $x$ on the shortest path from $u$ to $v$ with degree in $[2^i, 2^{i+1}]$. We show that $\delta_i$ provides the desired approximation. By definition of $S_i$, $x$ has at least one neighbor in $S_i$, denote this neighbor by $a$, see \autoref{fig:stretch_pic_i}.

\setlength{\intextsep}{4pt}
\begin{figure}[h]
\centering
\setlength{\abovecaptionskip}{0pt}
\setlength{\belowcaptionskip}{8pt}
\includegraphics[scale=0.6]{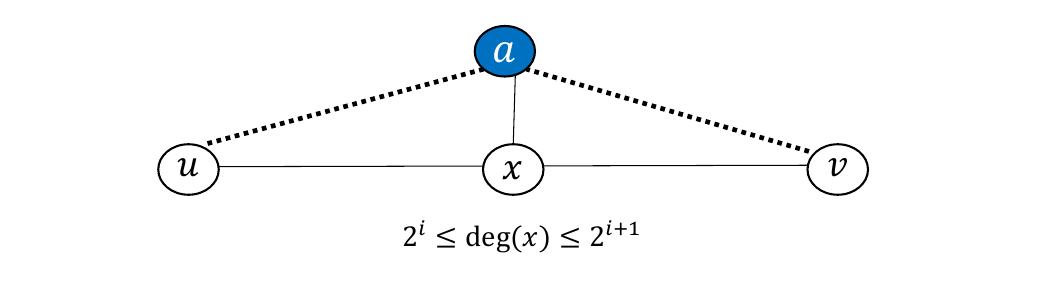}
 \caption{Illustration of the stretch analysis.}
\label{fig:stretch_pic_i}
\end{figure}

Now consider the distance from $u$ to $v$ in $G_i$. Since $G_i$ contains the shortest path in $G$, we have $d_{G_i}(u,v)=d_G(u,v)$. Further we know that $d(u,a)+d(a,v) \leq d_{G_i}(u,v) +2$, since $a$ is neighboring $x$ on the shortest path between $u$ and $v$. In particular, this means that $\min_{a'\in S_i} d_{G_i}(u,a') +d_{G_i}(a',v) \leq d_{G_i}(u,v)+2=d_G(u,v)+2$. So we can focus on computing shortest paths through $S_i$. In the second step of the algorithm, we have computed the (exact) distances $S_i$: $d_{G_i}(a',y)$ for all $y\in V$, in particular to $u$ and $v$. 
To compute $\min_{a'\in S_i} d_{G_i}(u,a') +d_{G_i}(a',v)$, we use Algorithm~$\mathcal{B}(G_i',S_i)$. which then provides an estimate $\delta_i(u,v)$ for the distance from $u$ to $v$ through $S_i$ in $G_i$. 
Hence, the approximation factors aggregate as follows:
\begin{align*} 
    \delta_i(u,v) &\leq \text{mult}_B(d_{G_i}(u,a)+d_{G_i}(a,v))+\text{add}_B\\
    &\leq \text{mult}_B(d_{G_i}(u,x)+1+d_{G_i}(x,v)+1)+\text{add}_B\\
    &=\text{mult}_B\cdot d_{G_i}(u,v)+2\cdot \text{mult}_B+\text{add}_B\\
    &=\text{mult}_B\cdot d(u,v)+2\cdot \text{mult}_B+\text{add}_B.
\end{align*}
\end{enumerate}

If Algorithm $\mathcal A$ is correct with probability at least $1-p_A$, and Algorithm $\mathcal{B}$ is correct with probability at least $1-p_B$, then our algorithm is correct with probability at least $1-p_A-p_B\cdot \log n $. Note that in all applications below we have $p_A=0=p_B$, i.e., both algorithms are always correct.

\paragraph{Running time.}
We run algorithm $\mathcal A(G')$, where $G'$ has $n$ vertices and $O(n^{2-r})$ edges, hence this takes $\tau_A(n,n^{2-r})$ time. Then, for each $i$, we perform three steps. First, we compute a hitting set in $O(m+n)$ time. Second, we run Dijkstra from $|S_i|=\tilde O(\tfrac{n}{2^i})$ vertices on a graph with $O(n2^{i+1})$ edges, taking $\tilde O(n^2)$ time. Third, we run algorithm $\mathcal B(G_i',S_i)$, where we know every shortest path goes through $S_i$ with $|S_i|= O(\tfrac{n}{2^i}\log n)=\tilde O(n^{r})$, hence taking $\tau_B(n,\tilde O(n^r))$ time. 

Any probabilistic guarantees on the running time of Algorithm~$\mathcal{A}$ and $\mathcal{B}$ simply carry over.

\subsection{\texorpdfstring{$2$}{2}-Approximate APSP for Unweighted Graphs} \label{sc:unweighted}
We employ the theorem of the previous section to obtain a multiplicative $2$-approximation. Hereto we use the following result of Baswana and Kavitha~\cite{BaswanaK10}, which gives an efficient algorithm for sparse graphs.\footnote{We note that \cite{BaswanaK10} states their result with an \emph{expected} running time. We can easily make this `with high probability' as follows. We run the algorithm $\log n$ time, stopping whenever we exceed the running time we aim for by more than a factor $\log n$. By a Chernoff bound, at least one of them finishes within this time w.h.p.}

\begin{theorem}[\cite{BaswanaK10}]\label{thm:BK}
     There exists a randomized, combinatorial algorithm that, given an undirected graph with non-negative edge weights $G=(V,E)$, computes $2$-approximate APSP. With high probability, the algorithm takes $\tilde O(n^2+m\sqrt{n})$ time. 
\end{theorem}

The algorithm of \autoref{thm:BK} can be retrieved from our more general algorithm for weighted graphs in \autoref{sc:weighted}.

\ThmTwoApx*
\begin{proof}
    By \autoref{lm:2eps_to_2}, it is sufficient to provide a $(2+\epsilon)$-approximation, given that the running time only depends polynomially on $1/\epsilon$. We utilize \autoref{thm:sparse_APSP} and for that we need to specify which algorithms to use for $\mathcal A$ and $\mathcal B$, and analyze the tradeoffs on the approximation and running time. 
     
    For algorithm $\mathcal A$ we use \autoref{thm:BK} (\cite{BaswanaK10}) on a graph with $n$ vertices and $n^{2-r}$ edges, which results in a $2$-approximation in $\tilde O(n^2+n^{5/2-r})$ time w.h.p. 
    Algorithm $\mathcal B$ has to provide shortest $u-v$ paths, that are minimal among all paths of the form $u-x-v$, for $x$ in a given set $S$. We use rectangular matrix multiplication for this: multiplying the $V\times S$ with the $S\times V$ edge weight matrices gives exactly the paths of length 2. We obtain $(1+\epsilon/2)$-approximate shortest paths through a set of $\tilde O(n^r)$ vertices in time $\tilde O(m^{1+o(1)}+n^{\omega(r)}(1/\epsilon)^{O(1)})$~\cite{Zwick02} (see also \autoref{thm:RMM_Zwick}). 

    By \autoref{thm:sparse_APSP}, for each pair of vertices we either obtain a stretch of $(2,0)$, or a stretch of $(1+\epsilon/2,(1+\epsilon/2)\cdot 2)$. 
    Note that for any $d(u,v)\geq 2$ we have  $(1+\epsilon/2)d(u,v)+(1+\epsilon/2)\cdot 2 \leq (1+\epsilon)d(u,v)+2$, so we have a $(1+\epsilon, 2)$-approximation, hence a $(2+\epsilon)$-approximation in total. 

    Also by \autoref{thm:sparse_APSP}, w.h.p.\ we obtain a running time of $\tilde O(n^2) +\tilde O(n^2+n^{5/2-r} )+ \tilde O(m^{1+o(1)}+n^{\omega(r)}(1/\epsilon)^{O(1)}))= \tilde O((n^{5/2-r}+n^{\omega(r)})(1/\epsilon)^{O(1)}) = \tilde O(n^{2.03184039}(1/\epsilon)^{O(1)})$ for $r=0.46815961$, where we use \cite{balancer} to balance the terms. 
\end{proof}

\subsection{\texorpdfstring{$2$}{2}-Approximate Combinatorial APSP for Unweighted Graphs} \label{sec:combinatotial}
The algorithm of the previous section uses matrix multiplication as a subroutine. In this section, we present a simple combinatorial algorithm for the same problem, matching the very recent result by Roditty~\cite{Roditty23}. 

\begin{restatable}{theorem}{ThmTwoApxComb}\label{thm:two_apx_comb}
There exists a \emph{combinatorial} algorithm that, given an unweighted, undirected graph $G=(V,E)$, computes $2$-approximate APSP. With high probability, the algorithm takes $\tilde O(n^{2.25})$  time. 
\end{restatable}
\begin{proof}
Again, we use \autoref{thm:sparse_APSP}. For algorithm $\mathcal A$ we use \autoref{thm:BK} to obtain a $2$-approximation in $\tilde O(n^2+n^{5/2-r})$ time w.h.p. 
Algorithm $\mathcal B$ has to provide shortest $u-v$ paths, that are minimal among all paths of the form $u-x-v$, for $x$ in a given set $S$. In particular, we can consider the graph $(V,S\times V)$, which has $n\cdot \tilde O(n^r)=\tilde O(n^{1+r})$ edges. Running Dijkstra from each node gives the (exact) result in $\tilde O(n^{2+r})$ time.

By \autoref{thm:sparse_APSP}, for each pair of vertices we either obtain a stretch of $(2,0)$, or a stretch of $(0,2)$, hence a $(2,0)$-approximation in total. 

Also by \autoref{thm:sparse_APSP}, with high probability, we obtain a running time of $\tilde O(n^{2.5-r}+n^{2+r})=\tilde O(n^{2.25})$, for $r=0.25$.

Since both \autoref{thm:BK} and Dijkstra are combinatorial, the final result is combinatorial.
\end{proof}

\subsection{\texorpdfstring{$(1+\epsilon,k)$}{(1+epsilon,k)}-Approximate APSP for Unweighted Graphs} \label{sec:near_additive}
As a third application of \autoref{thm:sparse_APSP}, we give an algorithm for computing $(1+\epsilon,k)$-approximate APSP for even $k\geq 2$. Hereto, we use the following result of Dor, Halperin, and Zwick~\cite{DHZ00}, which gives an efficient algorithm for sparse graphs.

\begin{theorem}[\cite{DHZ00}]\label{thm:DHZ}
\begin{enumerate}[i)]
    \item There exists a deterministic algorithm that, given an unweighted, undirected graph $G=(V,E)$, computes $+2$-approximate APSP in  $\tilde O(\min\{n^{3/2}m^{1/2},n^{7/3})$ time.
    \item There exists a deterministic algorithm that, given an unweighted, undirected graph $G=(V,E)$ and an even integer $k\geq 4$, computes $+k$-approximate APSP in $\tilde O(\min\{ n^{2-\tfrac{2}{k+2}}m^{\tfrac{2}{k+2}},n^{2+\tfrac{2}{3k-2}}\})$ time.  
\end{enumerate}

\end{theorem}
In particular, this gives a $+\log n$-approximation in $\tilde O(n^2)$ time. Also note that $n^{3/2}m^{1/2}= n^{2-\tfrac{2}{k+2}}m^{\tfrac{2}{k+2}}$ for $k=2$. So we can also say that for even $k\geq 2$ there exists a $+k$-approximate APSP algorithm that runs in $n^{2-\tfrac{2}{k+2}}m^{\tfrac{2}{k+2}}$ time.

\ThmNearAdditive*
\begin{proof}
Again, we use \autoref{thm:sparse_APSP}.
For algorithm $\mathcal A$, we use the result of \autoref{thm:DHZ} for sparse graphs. We apply it to a graph with $n$ vertices and $n^{2-r}$ edges, which leads to the running time 
$\tilde O(n^{2+(1-r)\tfrac{2}{k+2}})$. 

For algorithm $\mathcal B$, we use $(1+\epsilon/2)$-approximate rectangular matrix multiplication in time $\tilde O(n^{\omega (r)})(1/\epsilon)^{O(1)}$.

By \autoref{thm:sparse_APSP}, for each pair of vertices we either obtain a stretch of $(0,k)$, or a stretch of $(1+\epsilon/2,(1+\epsilon/2)2)$,
Note that for any $d(u,v)\geq 2$ we have $(1+\epsilon/2)d(u,v)+(1+\epsilon/2)\cdot 2 \leq (1+\epsilon)d(u,v)+2$, so we have a $(1+\epsilon, 2)$-approximation
hence a $(1+\epsilon,k)$-approximation in total. 

Also by \autoref{thm:sparse_APSP}, we obtain a running time of $\tilde O(n^{2+(1-r)\tfrac{2}{k+2}}+n^{\omega(r)}(1/\epsilon))$, for any choice of $r\in [0,1]$~\cite{Zwick02} (see also \autoref{thm:RMM_Zwick}).

Since both \cite{DHZ00} and \cite{Zwick02} are deterministic, the final result is deterministic. 
\end{proof}

To give optimal results, we pick $k$ as a function of $r$. Since there is no closed form for (the state of the art of) $\omega(r)$, we balance it for specific $k$ using~\cite{balancer}. 
\begin{itemize}
    \item[$k=2$:] we have $\tilde O(n^{2+(1-r)/2}+n^{\omega(r)}(1/\epsilon))= \tilde O(n^{2.15195331}/\epsilon)$, for $r= 0.69609339$.
    \item[$k=4$:] we have $\tilde O(n^{2+(1-r)/3}+n^{\omega(r)}(1/\epsilon))= \tilde O(n^{2.11900756}/\epsilon)$, for $r= 0.64297733$.
    \item[$k=6$:] we have $\tilde O(n^{2+(1-r)/4}+n^{\omega(r)}(1/\epsilon))= \tilde O(n^{2.0981921}/\epsilon)$, for $r= 0.60723159$.
    \item[$k=8$:] we have $\tilde O(n^{2+(1-r)/5}+n^{\omega(r)}(1/\epsilon))= \tilde O(n^{2.08383115}/\epsilon)$, for $r=0.58084427$.
    \item[$k\geq 10$:] The exponent will go to $2$, as $k$ goes to $\log n$. However, for $k\geq 10$, we do not improve upon the results of Dor, Halperin, and Zwick~\cite{DHZ00}.
\end{itemize}

\begin{table}[ht]
\begin{center}
\begin{tabular}{|c|c|c|c|} 
\hline
    $k$ & This work & Previous results (for $m=n^2$) \\ \hline
    $2$ & $n^{2.152}$ &  $n^{2.24+o(1)}$ \cite{berman2007faster} \\  \hline
    $4$ & $n^{2.119}$ &  $n^{2.2}$ \cite{DHZ00} \\  \hline
    $6$ & $n^{2.098}$ &  $n^{2.125}$ \cite{DHZ00} \\ \hline
    $8$ & $n^{2.084}$ &  $n^{2.091}$ \cite{DHZ00} \\ \hline
\end{tabular}
\end{center}
\caption{Comparison of our $(1+\epsilon,k)$-approximation and prior work.}
\label{table_near_additive}
\end{table}


\paragraph{Discussion.}
For future work, a further possibility would be to use a $+k$-approximation for Algorithm~$\mathcal B$ in \autoref{thm:sparse_APSP}, which in total gives a $+(k+2)$-approximation. However, that requires an efficient $+k$-approximation for MSSP, to the best of our knowledge no such algorithm exists at this time. Note that the $+k$ algorithm of \cite{DHZ00} is APSP and does not give a faster MSSP. Hence reducing $+(k+2)$ to $+k$ with that algorithm is only \emph{slower}.

\section{Weighted \texorpdfstring{$(2+\epsilon)$}{(2+epsilon)}-Approximate APSP}\label{sc:weighted}
The techniques of the previous section do not generalize well to the weighted setting. In this section we present an alternative approach for weighted graphs. 
First, we review some standard definitions and results on bunches and clusters in \autoref{sc:bunches}. Then we continue with two different approaches to some subroutines to obtain efficient algorithms for sparse and dense graphs, in \autoref{sc:static_distance_oracles} and \autoref{sc:weighted_dense} respectively. Finally, we show that we can generalize the latter result to obtain faster algorithms in a wider density range in \autoref{sc:parBK}.
Moreover, in \autoref{app:static2W}, we provide a faster algorithm for a $(2,W)$-approximation. This can be seen as an adaptation of \autoref{sc:static_distance_oracles}.

\subsection{Bunches and Clusters}\label{sc:bunches} 
We use the concepts bunches and clusters as defined by Thorup and Zwick~\cite{TZ2005}. Given a parameter $p\in[\tfrac{1}{n},1]$, called the \emph{cluster sampling rate}, we define $S$ to be a set of sampled vertices, where each vertex is part of $S$ with independent probability $p$. The \emph{pivot} $p(v)$ of a vertex $v$ is defined as the closest vertex $s$ in $S$, i.e., $p(v):= \arg \min_{s\in S} d(s,v)$. With equal distances, we can break ties arbitrarily. Now we define the bunch $B(u)$ of $u$ by $B(u):= \{v\in V : d(u,v) < d(u,p(u) \}$. If the set $S$ is not clear from context, we write $B(u,S)$. Clusters are the inverse bunches: $C(v) := \{u \in V : d(u,v) < d(u,p(u)\}=\{ u\in V: v\in B(u)\}$. By our random choice of $S$, we have with high probability that $|B(u)|\leq O(\tfrac{\log n}{p})$. 

Thorup and Zwick~\cite{TZ2005} showed how to compute this efficiently, in $\tilde O(\frac{m}{p})$ time. It is immediate that the \emph{total} load of the bunches and clusters is the same: $\sum_{u\in V}|B(u)|=\sum_{v\in V} |C(v)|$. Hence the \emph{average} size of each cluster is $\tilde O(\tfrac{1}{p})$. However, the maximal load over all clusters can still be big. 
In a later work, Thorup and Zwick~\cite{TZ01} showed that by refining the set $S$, we can bound the bunch and cluster sizes simultaneously.

\begin{lemma}[\cite{TZ01},\cite{TZ2005}]\label{lm:computebunches}
    There exists an algorithm \textsc{ComputeBunches}($G,p$) that, given a weighted graph $G$ and a parameter $p\in [\tfrac{1}{n},1]$, computes 
    \begin{itemize}
        \item a set of vertices $S\subseteq V$ of size $\tilde O(pn)$;
        \item pivots $p(u)\in A$ such that $d(u,p(u))=d(u,S)$ for all $u\in V$;
        \item the bunches and clusters with respect to~$S$;
        \item distances $d(u,v)$ for all $u\in V$ and $v\in B(u)\cup\{p(u)\}$.
    \end{itemize} 
    With high probability, we have that both the bunches and clusters have size at most $\tilde O(\tfrac{1}{p})$. The algorithm runs in time $\tilde O(\tfrac{m}{p})$.
\end{lemma}


\subsection{Structure of the Algorithm}\label{sc:structure}

At a high-level, our algorithm is as follows. We compute a suitable set of centers with bounded bunch and cluster size using \autoref{lm:computebunches}. Then for each pair of vertices whose bunches are adjacent (we depict this `adjacent case' on the left in \autoref{fig:adjacent_again}) we store a distance estimate through the edge connecting them. This case also includes the case that the bunches overlap or even when $u\in B(v)$. In both cases there is still an edge connecting the bunches, it is just contained \emph{within} one of the bunches already. We show that we maintain the exact distance in the adjacent case, and otherwise we give a $2$-approximation by going through the pivot. We provide pseudo-code in \autoref{alg:2_APSP}.

\begin{algorithm}[H] \SetAlgoLined 
\KwIn{Unweighted graph $G=(V,E)$, parameter $p\in [\tfrac{1}{n},1]$}
\textsc{ComputeBunches}($G,p$)\label{line:1}\;
Compute $(1+\epsilon/2)$-approximate MSSP from $S$, denoted by $\delta_S$\label{line:2}\;
Compute $\delta_{\rm{adjacent}}(u,v) = \min\{ d(u,u')+w(u',v')+d(v',v) : \{u',v'\} \in E, u'\in B(u), v'\in B(v)\}$\label{line:3}\;
\ForEach{$u,v\in V$\label{line:4}}{
$\delta(u,v) \leftarrow \min\{d(u,p(u))+\delta_S(p(u),v), d(v,p(v))+\delta_S(p(v),u), \delta_{\rm{adjacent}}(u,v)\}$\label{line:5}}
\Return{$\delta$}
\caption{A $(2+\epsilon)$-approximate APSP algorithm for unweighted graphs} 
\label{alg:2_APSP}
\end{algorithm}

\begin{lemma}\label{lm:weighted_correctness}
    The distance estimate $\delta$ returned by \autoref{alg:2_APSP} satisfies $d(u,v)\leq \delta(u,v) \leq (2+\epsilon)d(u,v)$ for every $u,v\in V$. 
\end{lemma}
\begin{proof}
   First, note that $\delta_S(x,y)\geq d(x,y)$ for any $x,y\in V$, and all other distance estimates making up $\delta(u,v)$ correspond to actual paths in the graph, hence $d(u,v)\leq \delta(u,v)$. Next, let $\pi$ be the shortest path from $u$ to $v$. We distinguish two cases. 

   \begin{figure}
    \centering
    \includegraphics[width=\textwidth]{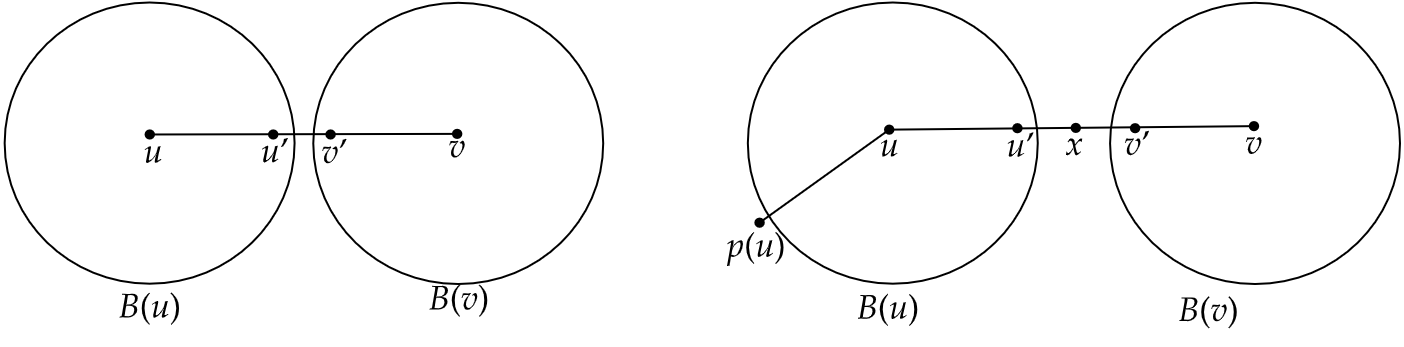}
    \caption{Two different possible interactions between the shortest path between $u$ and $v$, and the bunches of $u$ and $v$.}
    \label{fig:adjacent_again}
\end{figure}

\textbf{Case 1.} There exists $x\in \pi$ such that $x\notin (B(u)\cup B(v))$ (the right case in \autoref{fig:adjacent_again}). \\

Since $x\notin B(u)$, we have $d(u,x)\geq d(u,p(u))$, and since $x\notin B(v)$, we have $d(x,v)\geq d(v,p(v))$. Because $d(u,x)+d(x,v)=d(u,v)$, we have either $d(u,x)\leq \tfrac{d(u,v)}{2}$ or $d(x,v)\leq \tfrac{d(u,v)}{2}$. Without loss of generality, assume $d(u,x)\leq \tfrac{d(u,v)}{2}$. Then we have:
\begin{align*}
    \delta(u,v) &\leq d(u,p(u)) + \delta_S(v,p(u))\\
    &\leq d(u,p(u)) + (1+\epsilon/2)d(v,p(u))\\
    &\leq d(u,p(u)) +  (1+\epsilon/2)(d(u,p(u)) + d(u,v))\\
    &\leq (2+\epsilon/2)d(u,x) + (1+\epsilon/2)d(u,v)\\
    &\leq (2+\epsilon)d(u,v),
\end{align*}
where the third inequality holds by the triangle inequality. 

\textbf{Case 2.} There is no $x\in \pi$ such that $x\notin (B(u)\cup B(v))$ (the left case in \autoref{fig:adjacent_again}).\\
In other words, $\pi \subseteq (B(u)\cup B(v))$. This means there are vertices $u'\in B(u)$ and $v' \in B(v)$ such that $\{u',v'\}$ is an edge on the shortest path. Note that there is always at least one such edge, since $u\neq v$ and we allow $u'=u$ and $v'=v$. Since 
\begin{align*}
    \delta_{\rm{adjacent}}(u,v) 
    &= \min\{ d(u,u')+w(u',v')+d(v',v) : u'\in B(u) \text{ and }v'\in B(v)\},
\end{align*} we get $\delta(u,v) =\delta_{\rm{adjacent}}(u,v)= d(u,v)$ in this case. 

\end{proof}

In the following two sections, we give two different approaches to execute \autoref{alg:2_APSP}, giving efficient algorithms for sparse and dense graphs respectively. 

\subsection{An Efficient Distance Oracle for Sparse Graphs}\label{sc:static_distance_oracles}

In this section, we provide an efficient algorithm for $2$-approximate APSP in sparse graphs. We use the structure of the previous section and use combinatorial subroutines to obtain a combinatorial algorithm.
Note that as opposed to most of our other results (with an exception of \autoref{app:static2W}), this is a \emph{distance oracle} and not explicit APSP. For $m=\tilde O(n)$, our bounds match the conditional lower bound of $\tilde \Omega(m^{5/3})$ preprocessing time for $2$-approximations, conditional on the set intersection conjecture~\cite{PatrascuRT12} or the 3-SUM conjecture~\cite{aboud2022stronger}, where \cite{PatrascuRT12} shows the stronger $\tilde \Omega(m^{5/3})$ \emph{space} lower bound. 




\thmstatic*
\begin{proof}
\emph{Algorithm details and running time.}
Let $p\in [\tfrac{1}{n},1]$ be a parameter to be set later. We use \autoref{alg:2_APSP}, below we specify certain steps and the running time of this algorithm. 
\begin{itemize}
    \item \autoref{line:1} takes $\tilde O(\tfrac{m}{p})$ time by Lemma~\ref{lm:computebunches}. 
    \item For \autoref{line:2}, we use Dijkstra in $O(m|S|)$ time to obtain exact shortest paths ($\epsilon=0$). By \autoref{lm:computebunches} we have that $|S|=\tilde O(pn)$, hence this takes $\tilde O(pnm)$ time.
    \item By definition of bunches and clusters we have $u'\in B(u)\iff u\in C(u')$, so we can compute \autoref{line:3} as follows: for all $\{u',v'\} \in E$, for all $u\in C(u')$, and for all $v\in C(v')$:
\begin{enumerate}
    \item Initialize $\delta_{\rm{adjacent}}(u,v) \leftarrow d(u,u')+w(u',v')+d(v',v)$ if no such entry exists. \label{step:adj_init}
    \item Otherwise: $\delta_{\rm{adjacent}}(u,v)\leftarrow \min\{ \delta_{\rm{adjacent}}(u,v), d(u,u')+w(u',v')+d(v',v)\}$.\label{step:adj_update}
\end{enumerate}
    We note that Step~\ref{step:adj_init} and~\ref{step:adj_update} can both be done in constant time, so \autoref{line:3} takes
\[ \sum_{\{u',v'\}\in E} \sum_{u\in C(u')} \sum_{v\in C(v')} O(1) = \sum_{\{u',v'\}\in E} O(|C(u')|\cdot |C(v')|) = \sum_{\{u',v'\}\in E} \tilde O(\tfrac{1}{p^2}), \]
where the last equality holds by \autoref{lm:computebunches}. This means it takes $\tilde O(\tfrac{m}{p^2})$ time in total. 
\item Instead of executing the for-loop of \autoref{line:4}, we execute \autoref{line:5} in the query. This clearly takes constant time for a fixed pair $u,v\in V$. 
\end{itemize}
Together we obtain a running time of $\tilde O(\tfrac{m}{p}+pnm+\tfrac{m}{p^2})=\tilde O(pnm+\tfrac{m}{p^2})$. Balancing $pnm=\tfrac{m}{p^2}$ gives $p=n^{-1/3}$, so total running time $\tilde O(mn^{2/3}).$

\emph{Correctness.}
This holds by \autoref{lm:weighted_correctness}. As detailed above, we have $\epsilon=0$, obtaining a $2$-approximation.


\emph{Space.} We need $  O(|S|n)=\tilde O(pn^2)$ space for the distances from $S$, and $\tilde O(\tfrac{m}{p^2})$ space for the adjacent data structure. All other space requirements are clearly smaller. Inserting $p=n^{-1/3}$, we obtain total space requirement $\tilde O(mn^{2/3}+n^{5/3})=\tilde O(mn^{2/3})$.
\end{proof}

\subsection{\texorpdfstring{$(2+\epsilon)$}{2+epsilon}-Approximate APSP for Dense Graphs} \label{sc:weighted_dense}
In this section, we provide an efficient algorithm for $(2+\epsilon)$-approximate APSP for dense graphs. In this case, by dense we mean $m=n^2$. The algorithm of this section already improve on the state of the art for a wider range of $m$, but we defer the case of $m=o(n^2)$ to \autoref{sc:parBK}, where we obtain better results for this regime.  

\ThmTwoApxWeighted* 
\begin{proof}
Again, we will use a parameter $p\in [\tfrac{1}{n},1]$. For ease of notation, we let $r\in [0,1]$ such that $p=n^{r-1}$. We use \autoref{alg:2_APSP}, below we specify certain steps and the running time of this algorithm. 
\emph{Algorithm details and running time.}
\begin{itemize}
    \item \autoref{line:1} takes $\tilde O(\tfrac{m}{p})$ time by Lemma~\ref{lm:computebunches}. 
    \item For \autoref{line:2}, we use \cite{ElkinN22} this takes time $\tilde O(m^{1+o(1)}+n^{\omega(r)}(1/\epsilon)^{O(1)}\log W)$.
    \item We execute \autoref{line:3} as follows, in $\tilde O(\tfrac{n^2}{p})$ time:
\begin{enumerate}[a)]
        \item For all $v'\in V$, run Dijkstra on the graph $G_{v'}=(V,E_{v'})$, where $E_{v'}$ consists of all edges incident to $v'$, and for each vertex it contains edges connecting the vertex to the bunch. More formally, $E_{v'}:= \{\{v',x\}\in E : x\in V\}\cup \bigcup_{x\in V} \{ \{x,y\} : y\in B(x)\}$, with weights $w_{G_{v'}}(v',x):=w_G(v',x)$ and $w_{G_{v'}}(x,y):= d_G(x,y)$ respectively. 
        This second term has size $\tilde O(\tfrac{n}{p})$, since each bunch has size $\tilde O(\tfrac{1}{p})$. Hence computing SSSP with Dijkstra takes $\tilde O(\tfrac{n}{p}+n) = \tilde O(\tfrac{n}{p})$ time.\label{step:dijkstra_with_bunches1}
        \item For all $u\in V$, run Dijkstra on the graph $G_u':= (V,E_u')$, where $E_u'$ consists of the distances computed in the previous step and again the edges between vertices and their bunch. More formally, $E_u':= (\{u\}\times V)\cup \bigcup_{x\in V} \{ \{x,y\} : y\in B(x)\}$, with weights $w_{G_u'}(u,v'):=d_{G_{v'}}(v',u)$ and $w_{G_u'}(x,y):= d_G(x,y)$ respectively.
        For each $u$, the graph $G_u'$ has $\tilde O(\tfrac{n}{p})$ edges, hence computing shortest paths takes $\tilde O(\tfrac{n}{p})$ time using Dijkstra.  Denote the output of this step by $\delta_{\rm{adjacent}}(u,v)$.\label{step:dijkstra_with_bunches2}  
\end{enumerate}
\emph{Correctness.}
From Step~\ref{step:dijkstra_with_bunches1} we obtain an edge $\{u,v'\}$ for each $v'\in V$ such that there exists $u'\in B(u)$ with $\{u',v'\}\in E$. This edge has weight $w(u,v')=\min_{u'\in B(u)} d(u,u')+w(u',v')$. By Step~\ref{step:dijkstra_with_bunches2} we combine this with the shortest path from $v'$ to $v$ for $v'\in B(v)$. So in total we get $\delta_{\rm{adjacent}}(u,v)= \min\{ d(u,u')+w(u',v')+d(v',v) : \{u',v'\} \in E, u'\in B(u), v'\in B(v)\}$.
    \item The for-loop of \autoref{line:4} takes $n^2$ time.
\end{itemize}
In total we have $\tilde O(\tfrac{n^2}{p}+n^{\omega(r)}(1/\epsilon)^{O(1)}\log W)=\tilde O(n^{3-r}+n^{\omega(r)}(1/\epsilon)^{O(1)}\log W)$. For $r=0.7868671417236328$, we get $\tilde O(n^{2.21313612}(1/\epsilon)^{O(1)}\log W)=O(n^{2.214}(1/\epsilon)^{O(1)}\log W)$, using \cite{balancer}.

    \emph{Correctness.} By \autoref{lm:weighted_correctness} we obtain $(2+\epsilon)$-approximate APSP. 
\end{proof}


\subsection{A Parameterized APSP Algorithm for Weighted Graphs}\label{sc:parBK}
This section generalizes the $2$-approximation of Baswana and Kavitha~\cite{BaswanaK10}, such that there is a parameter in the running time controlling from how many sources we need to compute shortest paths. We then use fast matrix multiplication results to compute MSSP~\cite{ElkinN22} to do this part efficiently, and balance the parameters. 
We follow the algorithms and proofs of \cite{BaswanaK10}, making adjustments where necessary. 

We start by creating a hierarchy of $k= (1-r)\log n$ sets: $V= S_0 \supseteq S_1 \supseteq \dots \supseteq S_k$. We refer to this as an \emph{$r$-hierarchy}. For notational purposes we also have a set $S_{k+1}$, which we define to be the empty set: $S_{k+1}=\emptyset$.
We create these by starting with a hierarchy created by subsampling: we start with all vertices $S_0:=V$ and each subsequent $S_i'\subseteq S_{i-1}'$ is created by selecting each element of $S_{i-1}'$ with probability $\tfrac{1}{2}$. Now $|S_k'| =\tilde O(n \cdot (\tfrac{1}{2})^{(1-r)\log n})=\tilde O(n^r)$ w.h.p. This creates $V= S_0' \supseteq S_1' \supseteq \dots \supseteq S_{k}'$. Now we use \autoref{lm:computebunches} to compute a set of pivots $S$ of size $|S|=\tilde O(n^r)$, and we set $S_i=S_i'\cup S$. By adding the set $S$, we guarantee that the size of each cluster is bounded by at most $\tilde O(n^{1-r})$, by \autoref{lm:computebunches}.
We note that w.h.p.\ $|S_i|\leq |S_i'|+|S|= \tilde O(\tfrac{n}{2^i})+\tilde O(n^r)=\tilde O(\tfrac{n}{2^i})$. In particular, the last set of sources $S_k$ as size $|S_k| \leq |S_k'|+|S|=\tilde O (n^r)$. 
Using \autoref{lm:computebunches}, it is easy to see that we can compute an $r$-hierarchy in $\tilde O(mn^{1-r})$ time w.h.p., including pivots, bunches and clusters for each $S_i$.


Later, we compute shortest paths from the last level $S_k$, and show that either this gives a $2$-approximation, or that we can obtain a $2$-approximation through the lower levels. The latter is done with \autoref{alg:BK_scheme}. Here, for each level, we compute shortest paths from $S_i$ in a sparser graph, where the distances to some essential vertices equal the distances in the original graph. This suffices for correctness: on at least one of the levels, a distance estimate through the pivot of that level gives a $2$-approximation, see \autoref{lm:BK_scheme_correctness}.

\emph{Notation.} We denote $d(v,A)$ for the distance from a vertex $v\in V$ to a set $A\subseteq V$, i.e., $d(v,A):= \min_{u\in A} d(v,u)$. Next, we define $E_A(v)$ as the set of edges with weight less than $d(v,A)$, i.e., $E_A(v):= \{ \{u,v\} \in E : w(u,v) \leq d(v,A)\}$. And finally, we define $E_A:=\sum_{v\in V} E_A(v)$.
We recall that we use $B(u,A)$ to denote the bunch of $u$ w.r.t.\ some set of pivots $A$. 

\begin{algorithm}[H] \SetAlgoLined 
 \KwIn{An $r$-hierarchy $V= S_0 \supseteq S_1 \supseteq \dots \supseteq S_k \supseteq S_{k+1}= \emptyset$, for $k=(1-r)\log n$} 
\ForEach{$i\in [0,k]$ and $u\in V$}{
Compute pivot $p_i(u)$ of $u$ w.r.t.\ $S_i$\label{line:bk_scheme_pivot}\;
$\delta(u,p_i(u)) \leftarrow d(u,p_i(u))$}
\ForEach{$u\in V\setminus S_k$}{
Compute $B(u,S_k)$\label{line:bk_scheme_bunches}\;
\ForEach{$x \in B(u,S_k)$ and $0\leq i\leq k$}{
\ForEach{neighbor $y$ of $x$}{
$\delta(p_i(u),y) \leftarrow \min\{\delta(p_i(u),y), \delta(u,p_i(u))+\delta(u,x)+w(x,y)\}$ \label{line:bk_scheme_assigny}
}
}
}
\ForEach{$i\in [0,k-1]$ and $s\in S_i$}{
Run Dijkstra from $s$ on $(V,E_{S_{i+1}}\cup \{s\}\times V)$, with $w(s,v)=\delta(s,v)$, and update $\delta(s,v)$ for all $v\in V$ accordingly \label{line:bk_scheme_dijsktra}
}
\Return{$\delta$}
\caption{A subroutine for computing APSP~\cite[Algorithm 8]{BaswanaK10}} 
\label{alg:BK_scheme}
\end{algorithm}

Independent of our choice of $r$, this algorithm ensures a $2$-approximation for certain vertices. 

\begin{restatable}{lemma}{LmBKSchemeCorrectness}
[\cite{BaswanaK10} Theorem 7.3]\label{lm:BK_scheme_correctness}
    Let $u,v\in V$ be any two vertices, and let $\delta$ be the output of \autoref{alg:BK_scheme}. If $d(u,p_k(u))+d(v,p_k(v))>d(u,v)$, then  
    \begin{align*}
        \min_{0\leq i\leq k} \{\delta(u,p_i(u)) +\delta(p_i(u),v), \delta(v,p_i(v))+\delta(p_i(v),u)\} \leq 2d(u,v).
    \end{align*}
    This holds for any choice of $r\in[0,1]$.
\end{restatable}

The running time however does depend on $r$. By computing shortest paths on the sparser graph, we avoid the expensive computation of shortest paths from $S_i$ to all of $V$. Instead, for each level, we require $\tilde O(mn^{1-r})$ time to construct the sparser graph, and $\tilde O(n^2)$ time to compute shortest paths from $S_i$. 

\begin{restatable}{lemma}{LmBKSchemeRuntime}[\cite{BaswanaK10} Lemma 7.4]\label{lm:BK_scheme_time}
    Given $r\in [0,1]$, \autoref{alg:BK_scheme} takes $\tilde O(mn^{1-r}+n^2)$ time w.h.p. 
\end{restatable}

Next, we combine this subroutine with shortest paths from the last level to obtain a $2$-approximation, see \autoref{alg:BK_APSP}. This corresponds to Algorithm 9 of \cite{BaswanaK10}, where we parameterize the hierarchy.

\begin{algorithm}[H] \SetAlgoLined 
$k\leftarrow (1-r)\log n$\;
Compute an $r$-hierarchy $V= S_0 \supseteq S_1 \supseteq \dots \supseteq S_{k} \supseteq S_{k+1}= \emptyset$\;
Run \autoref{alg:BK_scheme} w.r.t.\ $S_0, S_1, \dots, S_k, S_{k+1}= \emptyset$\;
Compute $(1+\epsilon/2)$-approximate MSSP from $S_k$ in $(V,E)$ \label{line:alg_BK_MSSP}

\ForEach{$u,v\in V$\label{line:alg_BK_for}}{
\If{$v\notin B(u,S_k)$ and $u\notin B(v,S_k)$}{
$\delta(u,v) \leftarrow \min_{0\leq i \leq k}\{\delta(u,p_i(u))+\delta(p_i(u),v), \delta(v,p_i(v))+\delta(p_i(v),u)\}$}
}
\Return{$\delta$}
\caption{A $2$-approximate APSP algorithm} 
\label{alg:BK_APSP}
\end{algorithm}

We show that independent of our choice for $r$, this gives a $2$-approximation. In case $\epsilon=0$, this is \cite[Lemma 7.5]{BaswanaK10}, we adapt this proof to allow for approximate shortest paths.

\begin{lemma}\label{lm:BK_correctness}
    \autoref{alg:BK_APSP} computes $(2+\epsilon)$-approximate APSP for any choice of $r\in[0,1]$.  
\end{lemma}
\begin{proof}
    First of all, if $u\in B(v,S_k)$ or $v\in B(u,S_k)$ the exact distance is known by $\delta(u,v)$. So for the rest of the proof we assume otherwise. Now if $d(u,p_k(u))+d(v,p_k(v))>d(u,v)$, then we obtain a $2$-approximation by \autoref{lm:BK_scheme_correctness}. 

    We are left with the case that $d(u,p_k(u))+d(v,p_k(v))\leq d(u,v)$. Without loss of generality, let $d(u,p_k(u))\leq d(u,v)/2$. By \autoref{alg:BK_scheme}, we have $\delta(u,p_k(u))=d(u,p_k(u))$. Furthermore, by \autoref{line:alg_BK_MSSP}, we have $\delta(p_k(u),v)\leq (1+\epsilon/2)d(p_k(u),v)$. In total we obtain by \autoref{line:alg_BK_for} and the triangle inequality that 
    \begin{align*}
        \delta(u,v) &\leq \delta(u,p_k(u))+\delta(p_k(u),v)\leq d(u,p_k(u))+(1+\epsilon/2)d(p_k(u),v) \\
        &\leq d(u,p_k(u))+(1+\epsilon/2)d(u,p_k(u))+(1+\epsilon/2)d(u,v) \leq (2+\epsilon)d(u,v).
    \end{align*}

   Since all distance estimates $\delta(u,v)$ correspond to paths in the graph, we trivially have $d(u,v)\leq \delta(u,v)$.
\end{proof}

Next, we show how the running time depends on $r$. 

\begin{lemma}\label{lm:BK_time}
    For $r\in[0,1]$, \autoref{alg:BK_APSP} takes $\tilde O(n^2+mn^{1-r}+T(\tilde O(n^{r}))$ time w.h.p., where $T(s)$ is the time to compute $(1+\epsilon)$-approximate MSSP from $s$ sources in a graph with $n$ vertices and $m$ edges. 
\end{lemma}
\begin{proof}
    We can compute an $r$-hierarchy in $\tilde O(mn^{1-r})$ time w.h.p.\ (follows directly from the definition and \autoref{lm:computebunches}). \autoref{alg:BK_scheme} takes $\tilde O(mn^{1-r}+n^2)$ time (\autoref{lm:BK_scheme_time}). Next, in \autoref{line:alg_BK_MSSP}, we need to compute MSSP from $|S_k|=\tilde O(n^r)$, for which we denote the running time as $T(\tilde O(n^{r}))$. Finally, the for-loop of \autoref{line:alg_BK_for} takes $O(n^2k)=\tilde O(n^2)$ time. Adding all running times, we obtain $\tilde O(n^2+mn^{1-r}+T(\tilde O(n^{r}))$ time w.h.p.
\end{proof}

Together \autoref{lm:BK_correctness} and \autoref{lm:BK_time} give \autoref{thm:ParBK}. 

\ThmParBK*

\paragraph{\texorpdfstring{$(2+\epsilon)$}{2+epsilon}-Approximate APSP for Weighted Graphs}
Baswana and Kavitha~\cite{BaswanaK10} proceed by setting $r = 1/2$ (or equivalently $p=1/\sqrt n$). For the MSSP computations they use Dijkstra (hence $\epsilon=0$) in $\tilde O(n^{1-r}m)=\tilde O(m\sqrt n)$ time, see also \autoref{thm:BK}. Instead, we keep $r$ as a parameter, and use fast matrix multiplication to obtain $(1+\epsilon)$-approximate MSSP. 

\ThmTwoApxWeightedPar*
\begin{proof}
    This follows directly from \autoref{thm:ParBK}, combined with the $(1+\epsilon)$-approximate MSSP algorithm of \cite{ElkinN22} (see \autoref{thm:EN_MSSP}). 
\end{proof}

For dense graphs, i.e., $m=n^2$, we can balance the terms using \cite{balancer}. If we do so, we recover \autoref{thm:two_apx_wghted}. Results for other densities are obtained in a similar fashion, see \autoref{table_weighted} for the results. 

\ThmTwoApxWeighted*
\begin{proof}
    For $m=n^2$, w.h.p.\ the running time of \autoref{thm:two_apx_wghted_par} becomes $\tilde O(n^{3-r}+n^{\omega(r)}(1/\epsilon)^{O(1)}\log W)=\tilde O(n^{2.21313612}(1/\epsilon)^{O(1)}\log W)$ for $r=0.78686388$.
\end{proof}


\printbibliography[heading=bibintoc] 

\newpage
\appendix
\section{Utilizing Recent Improvements on Rectangular Matrix Multiplication}\label{app:newRMM}
Rectangular matrix multiplication is an active research field, with the bounds on $\omega(r)$ being improved in recent years~\cite{gall2023faster, GallU18, le2012faster}. Throughout this paper, we used \cite{GallU18}, the last published paper on the topic, for the sake of replicability. However, more recent, concurrent work by Vassilevska Williams, Xu, Xu, and Zhou~\cite{VassilevskaWXXZ23} gives better bounds. In this section, we detail how this affects our running times.

Our result for $2$-approximate APSP in unweighted graphs, \autoref{thm:two_apx}, has running time $\tilde O(n^{2.5-r}+n^{\omega(r)})=O(n^{2.031062336})$, for $r=0.4689376644$. 

Our results for $(2+\epsilon)$-approximate APSP in weighted graphs, \autoref{thm:two_apx_wghted_par}, are given in \autoref{table:weighted_new}.
\begin{table}[h]
\begin{center}
\begin{tabular}{|c|c|c|c|} 
\hline
    $m$ & Runnning time & Using \cite{VassilevskaWXXZ23} for $\omega(r)$ & with $r$\\ \hline
    $n^{1.4}$ & $n^{2.4-r}+n^{\omega(r)}$ &  $n^{2.008199835}$ & $0.3918001650$ \\  \hline
    $n^{1.5}$ & $n^{2.5-r}+n^{\omega(r)}$ &  $n^{2.031062336}$ & $0.4689376644$ \\  \hline
    $n^{1.6}$ & $n^{2.6-r}+n^{\omega(r)}$ &  $n^{2.061029532}$ & $0.5389704676$ \\  \hline
    $n^{1.7}$ & $n^{2.7-r}+n^{\omega(r)}$ &  $n^{2.095342149}$ & $0.6046578512$ \\  \hline
    $n^{1.8}$ & $n^{2.8-r}+n^{\omega(r)}$ &  $n^{2.132619229}$ & $0.6673807708$ \\  \hline
    $n^{1.9}$ & $n^{2.9-r}+n^{\omega(r)}$ &  $n^{2.171770761}$ & $0.7282292393$ \\  \hline
    $n^{2.0}$ & $n^{3-r}+n^{\omega(r)}$ &  $n^{2.212352011}$ & $0.7876479892$ \\  \hline
\end{tabular}
\end{center}
\caption{Our $(2+\epsilon)$-approximate APSP results (\autoref{thm:two_apx_wghted_par}) using~\cite{VassilevskaWXXZ23}, for $1/\epsilon=n^{o(1)}$.  \autoref{thm:two_apx_wghted_par} gives the fastest running time when $m\geq n^{1.544}$. For $m < n^{1.544}$, we do not improve on \cite{BaswanaK10}.}
\label{table:weighted_new}
\end{table}

Our results for near-additive APSP in unweighted graphs, \autoref{thm:near_additive}, are given in \autoref{table:near_additive_new}. 

\begin{table}[ht]
\begin{center}
\begin{tabular}{|c|c|c|c|} 
\hline
    $k$ & Runnning time & Using \cite{VassilevskaWXXZ23} for $\omega(r)$ & with $r$\\ \hline
    $2$ & $n^{2+(1-r)/2}+n^{\omega(r)}$ &  $n^{2.151353127}$ & $0.6972937458$ \\  \hline
    $4$ & $n^{2+(1-r)/3}+n^{\omega(r)}$ &  $n^{2.118511896}$ & $0.6444643104$ \\  \hline
    $6$ & $n^{2+(1-r)/4}+n^{\omega(r)}$ &  $n^{2.097785917}$ & $0.6088563325$ \\  \hline
    $8$ & $n^{2+(1-r)/5}+n^{\omega(r)}$ &  $n^{2.083460832}$ & $0.5826958380$ \\  \hline
\end{tabular}
\end{center}
\caption{Our $(1+\epsilon,k)$-approximate APSP results (\autoref{thm:near_additive}) using~\cite{VassilevskaWXXZ23}, for $1/\epsilon=n^{o(1)}$.}
\label{table:near_additive_new}
\end{table}

\section{\texorpdfstring{$(2,W_{u,v})$}{(2,W{u,v})}-Approximate APSP} \label{app:static2W}


In this section, we prove \autoref{thm:static_2W}, stated below, providing $(2,W_{u,v})$-approximate shortest paths. This is a generalized version of Baswana, Goyal, and Sen~\cite{BaswanaGS09}, who provide $(2,1)$-APSP in \emph{unweighted} graphs in $\tilde O(nm^{2/3}+n^2)$ time. We make three improvements in our generalization: 1) the algorithm allows for \emph{weighted} graphs, 2) it achieves subquadratic time for $m\leq n^{3/2}$, since it is a distance oracle rather than explicit APSP, and 3) we achieve a faster running time for $m>n^{3/2}$, by having a wider choice in parameters. 

The structure of the algorithm is similar to \autoref{sc:static_distance_oracles}. We show that for vertices whose bunches not overlap, the path through the pivot is actually a $(2,W_{u,v})$-approximation (even if the bunches are adjacent, top left case in \autoref{fig:overlap}). Then for bunches that do overlap (the bottom case in \autoref{fig:overlap}), we create a data structure to store these distances. 

   \begin{figure}
    \centering
    \includegraphics[width=\textwidth]{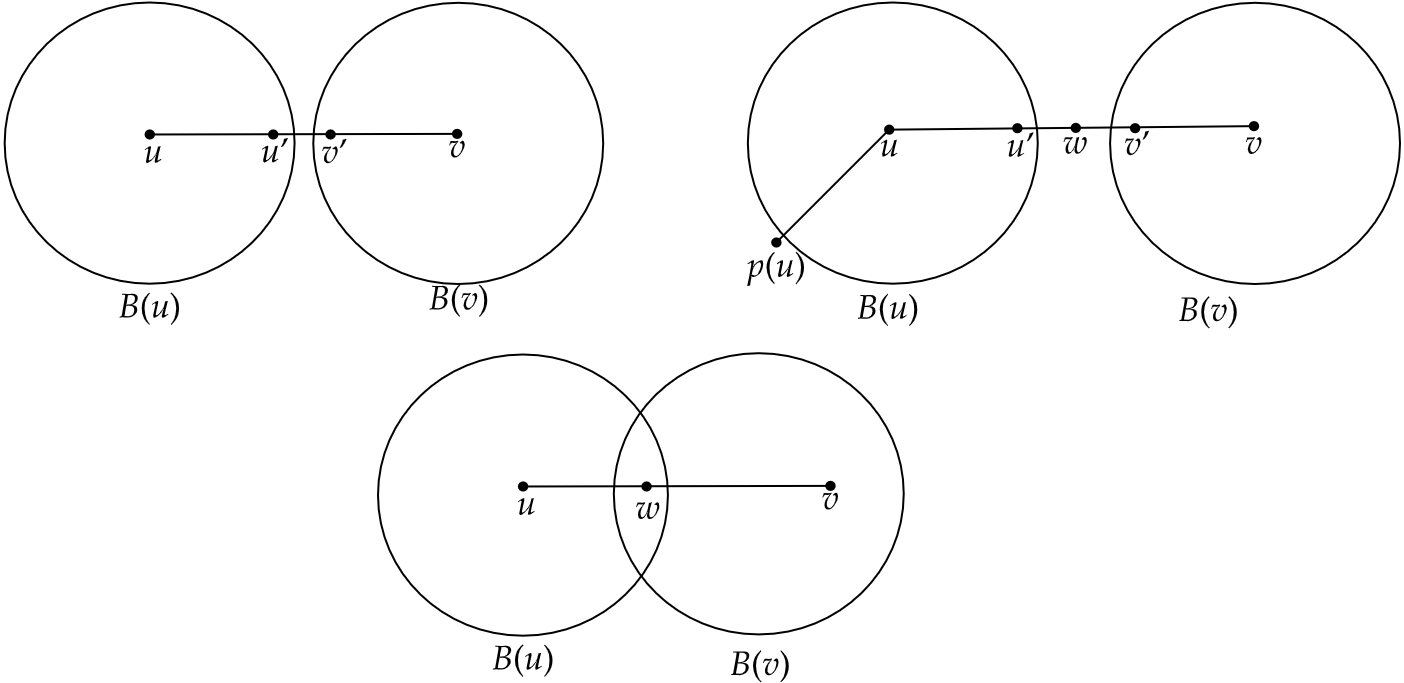}
    \caption{Three different possible interactions between the shortest path between $u$ and $v$, and the bunches of $u$ and $v$.}
    \label{fig:overlap}
\end{figure}


We note that this result is mostly interesting in the regime where $m=O(n^{3/2})$, where we obtain \emph{subquadratic} time and space. For denser graphs, we do not improve upon Baswana and Kavitha~\cite{BaswanaK10}, who provide a $(2,W_{u,v})$-approximation in $\tilde O(n^2)$ time and space.

\begin{restatable}{theorem}{thmstaticW}\label{thm:static_2W}
    Given a weighted graph $G$, we can compute a distance oracle that returns $(2,W_{u,v})$-approximate queries in constant time, where $W_{u,v}$ is the maximum weight on a shortest path from $u$ to $v$, with one of the following guarantees:
    \begin{itemize}
         \item preprocessing time $\tilde O(nm^{2/3})$ and uses space $\tilde O(nm^{2/3})$, when $m\leq n^{3/2}$, and
         \item preprocessing time $\tilde O(mn^{1/2})$ and uses space $O(n^2)$, when $m>n^{3/2}$,
    \end{itemize}
    or it has (optimizing for space) running time $\tilde O(mn^{2/3}) $ and $\tilde O(n^{5/3})$ space. 
\end{restatable}

\begin{proof}
\emph{Algorithm.} 
Let $p\in [\tfrac{1}{n},1]$ be a parameter, to be chosen later. 

\noindent\PreProcessing{$G,p$}:
\begin{enumerate}
    \item Run \textsc{ComputeBunches}($G,p$) \label{step:computebunchesW}
    \item Run Dijkstra for each vertex~$u\in A$ to compute all distances $d(u,v)$ for $u\in A, v\in V$. \label{step:dijkstraW}
    \item For $u\in V$, for $w\in B(u)$, for $v\in C(w)$: \label{step:overlap}
    \begin{enumerate}
        \item Initialize $\delta_{\rm{overlap}}(u,v) \leftarrow d(u,w)+d(w,v)$ if no such entry exists. \label{step:overlap_init}
        \item Otherwise: $\delta_{\rm{overlap}}(u,v)\leftarrow \min\{ \delta_{\rm{overlap}}(u,v), d(u,w)+d(w,v)\}$. \label{step:overlap_update}
    \end{enumerate}
\end{enumerate}

\noindent\Query{$u,v$}:\\
    \indent Output $\delta(u,v)$ to be the minimum of \label{step:cases_staticW}
    \begin{enumerate}[(a)]
         \item $\min\{ d(u,p(u))+d(v,p(u)),d(u,p(v))+d(v,p(v))\}$;\label{step:pivot_staticW} 
        \item $\delta_{\rm{overlap}}(u,v)$; \label{step:overlap_static}
    \end{enumerate}

\emph{Correctness.}
We will show that $\delta(u,v)$ gives a $(2,W_{u,v})$-approximation of $d(u,v)$. First, note that all distances making up $\delta(u,v)$ correspond to actual paths in the graph, hence $d(u,v)\leq \delta(u,v)$. Next, let $\pi$ be the shortest path from $u$ to $v$. We distinguish two cases. 

\textbf{Case 1.} There exists $w\in \pi$ such that $w\in B(u)\cap B(v)$ (the bottom case in \autoref{fig:overlap}).

We have that 
\begin{align*}
    \delta_{\rm{overlap}}(u,v) &= \min\{ d(u,x)+d(x,v) : x\in B(u) \text{ and } y\in C(x)\}\\
    &= \min\{ d(u,x)+d(x,v) : x\in B(u) \text{ and } x\in B(y)\}.
\end{align*}
In particular, this includes $x=w$, and $w$ is on the shortest path, so by Query Step~\ref{step:overlap_static}, we have that $\delta (u,v)\leq d(u,v)$. 

\textbf{Case 2.} There is no $w\in \pi$ such that $w\in B(u)\cap B(v)$.

This means there is either a vertex $w\in \pi$ such that $w \notin B(u)\cup B(v)$ (the top right case in \autoref{fig:overlap}), or there is an edge $\{u',v'\}$ on $\pi$ such that $u' \in B(u)\setminus B(v)$ and $v'\in B(v)\setminus B(u)$ (the top right case in \autoref{fig:overlap}). The first case gives a 2-approximation by the same reasoning as in the proof of \autoref{lm:weighted_correctness}, where we use exact shortest path, hence $\epsilon =0$. Here we only consider the second case. Since $v'\notin B(u)$, we have $d(u,v')\geq d(u,p(u))$, and since $u'\notin B(v)$, we have $d(u',v)\geq d(v,p(v))$. Combining this, we obtain $d(u,p(u))+d(v,p(v)) \leq d(u,v)+w(u',v')$. Without loss of generality, assume that $d(u,p(u)) \leq \tfrac{d(u,v)+w(u',v')}{2}$. By Query Step~\ref{step:pivot_staticW} we have:
\begin{align*}
    \delta(u,v) &\leq d(u,p(u)) + d(v,p(u))\\
    &\leq d(u,p(u)) + d(u,p(u)) + d(u,v)\\
    &\leq 2d(u,v)+w(u',v')\\
    &\leq 2d(u,v)+W_{u,v}.
\end{align*}

\emph{Running time.}
Step~\ref{step:computebunchesW} and~\ref{step:dijkstraW} take $\tilde O(\tfrac{m}{p})$ and $\tilde O(pnm)$ time respectively, see \autoref{lm:computebunches}. For Step~\ref{step:overlap}, notice that Step~\ref{step:overlap_init} and~\ref{step:overlap_update} both take constant time. So Step~\ref{step:overlap} takes total time
\begin{align*}
    \sum_{u\in V} \sum_{w\in B(u)} \sum_{v\in C(w)}O(1)= \tilde O(\tfrac{n}{p^2}),
\end{align*}
since $|B(u)|=\tilde O(\tfrac{1}{p})$ and $|C(w)|=\tilde O(\tfrac{1}{p})$, by \autoref{lm:computebunches}. We obtain total time $\tilde O(\tfrac{m}{p}+pnm+\tfrac{n}{p^2})$. We can balance this in three different ways:
\begin{itemize}
    \item $\tfrac{m}{p}=pnm$, which implies $p=n^{-1/2}$ and gives running time $\tilde O(mn^{1/2}+n^2)$.
    \item $\tfrac{m}{p}=\tfrac{n}{p^2}$, which implies $p=\tfrac{n}{m}$ and gives running time $\tilde O(\tfrac{m^2}{n}+n^2)$. 
    \item $pnm = \tfrac{n}{p^2}$, which implies $p=m^{-1/3}$ and gives running time $\tilde O(m^{4/3}+nm^{2/3})$.
\end{itemize}
Note that $\tilde O(mn^{1/2}+n^2)$ is always smaller than $\tilde O(\tfrac{m^2}{n}+n^2)$ , since $mn^{1/2} \leq \tfrac{m^2}{n}$ for $m\geq n^{3/2}$ and $mn^{1/2}\leq n^2$ for $m\leq n^{3/2}$.

Further we see that $\tilde O(m^{4/3}+nm^{2/3})$ is smaller than $\tilde O(mn^{1/2}+n^2)$ when $m\leq n^{3/2}$, since then $m^{4/3}\leq n^2$ and $mn^{1/2}\leq n^2$. 

Finally we notice that $nm^{2/3} \geq m^{4/3}$ when $m\leq n^{3/2}$, and that $mn^{1/2}\geq n^2$ for $m\geq n^{3/2}$. So our running time simplifies to 
\begin{itemize}
     \item $\tilde O(nm^{2/3})$, when $m\leq n^{3/2}$, and
     \item $\tilde O(mn^{1/2})$, when $m>n^{3/2}$.
\end{itemize}

\emph{Query time.} We note that we have computed $B(u)$ and $B(v)$, so checking if $v\in B(u)$ or $u\in B(v)$ can be done in constant time. Further, we have also computed $d(u,v)$ if $v\in B(u)$. For Query Step~\ref{step:pivot_staticW}, notice that we computed distances from $A$ to $V$, and for Query Step~\ref{step:overlap_static} we have already computed the value $ \delta_{\rm{overlap}}(u,v)$. So the whole query can be done in constant time. 

\emph{Space.} We need $  O(|A|n)=\tilde O(pn^2)$ space for the distances from $A$, and $\tilde O(\tfrac{n}{p^2})$ space for the overlap data structure. All other space requirements are clearly smaller. Using $p=n^{-1/2}$ gives space $O(n^2)$, using $p=m^{-1/3}$ gives space $\tilde O(nm^{2/3}+n^2m^{-1/3})= \tilde O(nm^{2/3})$. 

To optimize the space usage, we set $pn^2=\tfrac{n}{p^2}$, so $p=n^{-1/3}$, we obtain total space requirement $\tilde O(n^{5/3})=\tilde O(mn^{2/3})$. This gives running time $\tilde O(mn^{1/3}+mn^{2/3}+n^{5/3})=\tilde O(mn^{2/3})$.
\end{proof}

\newpage

\end{document}